\documentclass[11pt,a4paper,fleqn]{article}

\usepackage{amsmath,amssymb,amsthm,enumerate}

\newcommand{\p}{\partial}

\newcommand{\lsemioplus}{\mathbin{\mbox{$\lefteqn{\hspace{.77ex}\rule{.4pt}{1.2ex}}{\in}$}}}
\newcommand{\sgn}{\mathop{\rm sgn}\nolimits}
\newcommand{\spanindex}{{\mbox{\tiny$\langle\,\rangle$}}}

\newcommand{\todo}[1][\null]{\ensuremath{\clubsuit}}

\newcommand{\noprint}[1]{}
\newcommand{\checked}[1][\null]{\ensuremath{\boldsymbol{\surd}}}

\newtheorem{theorem}{Theorem}
\newtheorem{lemma}[theorem]{Lemma}
\newtheorem{corollary}[theorem]{Corollary}
\newtheorem{proposition}[theorem]{Proposition}
\newtheorem*{proposition*}{Proposition}
{\theoremstyle{definition}
\newtheorem{definition}[theorem]{Definition}

\newtheorem{remark}[theorem]{Remark}
}

\begin{document}

\par\noindent {\LARGE\bf
Algebraic method for group classification\\
of (1+1)-dimensional linear Schr\"odinger equations 
\par}

\vspace{5mm}\par\noindent{\large
C\'elestin Kurujyibwami$^{\dag}$, Peter Basarab-Horwath$^\ddag$ and Roman O. Popovych$^\S$
}\par\vspace{2mm}\par

\vspace{2mm}\par\noindent{\it
$^\dag$\,College of Science and Technology, University of Rwanda, P.O.\,Box: 3900, Kigali, Rwanda\\
$\phantom{^\dag}$\,Link\"oping University, 58183 Link\"oping, Sweden}\\
$\phantom{^\dag}$\,E-mail: celeku@yahoo.fr 

\vspace{2mm}\par\noindent{\it
$^\ddag$\,Link\"oping University, 58183 Link\"oping, Sweden}\\
$\phantom{^\ddag}$\,E-mail: pehor@mail.liu.se

\vspace{2mm}\par\noindent{\it
$^\S$\,Wolfgang Pauli Institute, Oskar-Morgenstern-Platz 1, A-1090 Vienna, Austria\\
$\phantom{^\S}$Institute of Mathematics of NAS of Ukraine, 3 Tereshchenkivs'ka Str., 01004 Kyiv, Ukraine}\\
$\phantom{^\S}$\,E-mail: rop@imath.kiev.ua

\vspace{6mm}\par\noindent\hspace*{10mm}\parbox{140mm}{\small
We carry out the complete group classification of the class 
of (1+1)-dimensional linear Schr\"odinger equations with complex-valued potentials. 
After introducing the notion of uniformly semi-normalized classes of differential equations, 
we compute the equivalence groupoid of the class under study and show that it is uniformly semi-normalized. 
More specifically, each admissible transformation in the class is 
the composition of a linear superposition transformation of the corresponding initial equation 
and an equivalence transformation of this class. 
This allows us to apply the new version of the algebraic method based on uniform semi-normalization
and reduce the group classification of the class under study
to the classification of low-dimensional appropriate subalgebras of the associated equivalence algebra. 
The partition into classification cases involves two integers
that characterize Lie symmetry extensions and are invariant with respect to equivalence transformations.}

\vspace{4mm}\par

\noprint{

Algebraic method for group classification of (1+1)-dimensional linear Schr\"odinger equations 

C\'elestin Kurujyibwami, Peter Basarab-Horwath and Roman O. Popovych

MSC
35Q41 (Primary) 35B06, 35A30 (Secondary)
35-XX   Partial differential equations 
  35Qxx  Equations of mathematical physics and other areas of application  
*   35Q41   Time-dependent Schrodinger equations, Dirac equations
    35Q40   PDEs in connection with quantum mechanics 
  35Bxx  Qualitative properties of solutions 
*   35B06   Symmetries, invariants, etc 
  35Axx  General topics 
*   35A30   Geometric theory, characteristics, transformations [See also 58J70, 58J72] 

  58Jxx  Partial differential equations on manifolds; differential operators [See also 32Wxx, 35-XX, 53Cxx] 
    58J70   Invariance and symmetry properties [See also 35A30] 


}

\section{Introduction}

A standard assumption of quantum mechanics requires that the Hamiltonian of a quantum system be Hermitian
since this guarantees that the energy spectrum is real 
and that the time evolution of the system is unitary and hence probability-preserving~\cite{Bender2007}.
For linear Shr\"odinger equations, this means that only equations with real-valued potentials 
are considered to be physically relevant. 
Since the above assumption is, unlike the other axioms of quantum mechanics, more mathematical than physical,
attempts at weakening or modifying the Hermitian property of Hamiltonians 
have recently been made by looking at so-called $\mathcal{PT}$-symmetric Hamiltonians 
\cite{Bender2004,Bender2007,Mostafazadeh2014}. 
Here $\mathcal P$ is the space reflection (or parity) operator 
and $\mathcal T$ is the time reversal operator.     
Some complex potentials are associated with non-Hermitian $\mathcal{PT}$-symmetric Hamiltonians. 
Non-Hermitian $\mathcal{PT}$-symmetric Hamiltonians have been used to describe (observable) phenomena 
in quantum mechanics, such as systems interacting  with electromagnetic fields, 
dissipative processes such as radioactive decay, 
the ground state of Bose systems of hard spheres and both bosonic and fermionic degrees of freedom. 
Other important applications of non-Hermitian $\mathcal{PT}$-symmetric Hamiltonians are to be found 
in scattering theory which include numerical investigations of various physical phenomena in optics, 
condensed matter physics, scalar wave equations (acoustical scattering) 
and Maxwell's equations (electromagnetic scattering), in quasi-exactly solvable Hamiltonians, 
complex crystals and quantum field theory \cite{Bender2004,Bender2007,Mostafazadeh2014}.
In general, however, the physical interpretation of linear Schr\"odinger equations 
with complex potentials is not completely clear.

The study of Lie symmetries of Schr\"odinger equations was begun in the early 1970's 
after the revival of Lie's classical methods (see for instance~\cite{Ovsiannikov1982}).  
Lie symmetries of the free (1+3)-dimensional Schr\"odinger equations 
were first considered in~\cite{Niederer1972}. 
Therein it was suggested to call the essential part of the maximal Lie symmetry group 
of the free Schr\"odinger equation the Schr\"odinger group.
In~\cite{Niederer1973a} it was noted that the results of~\cite{Niederer1972} 
could be extended directly to any number of space variables, 
and the isomorphism of the Lie symmetry groups of the Schr\"odinger equations 
of the $n$-dimensional harmonic oscillator 
and of the $n$-dimensional free fall 
to the symmetry group for the (1+$n$)-dimensional free Schr\"odinger equation 
was proved in~\cite{Niederer1973a,Niederer1973b,Niederer1974}. 
This gave a hint for the construction of point transformations connecting these equations. 
The problem of finding Lie symmetries of (1+$n$)-dimensional linear Schr\"odinger equations 
with real-valued potentials was considered in~\cite{Boyer1974,Niederer1974}. 
In particular, in~\cite{Niederer1974} 
the general ``potential-independent'' form of point symmetry transformations of these equations
was found under the a priori assumption of fibre preservation. 
The classifying equation involving both transformation components and the potential 
was derived and used to obtain an upper bound of dimensions of Lie symmetry groups admitted 
by linear Schr\"odinger equations. 
Then some static potentials of physical relevance were considered, 
including the harmonic oscillator, the free fall, the inverse square potential, 
the anisotropic harmonic oscillator 
and the time-dependent Kepler problem.
The case of arbitrary time-independent real-valued potential was studied in~\cite{Boyer1974}.     
Although it was claimed there that 
``the general solution and a complete list of such potentials and their symmetry groups are then
given for the cases $n=1,2,3$'', it is now considered that this list is not complete. 
Note that in the papers cited above, phase translations and amplitude scalings were ignored, 
which makes certain points inconsistent. 

Similar studies were carried out for the time-dependent Schr\"odinger equation 
for the two-dimensional harmonic oscillator and for the two- and three-dimensional hydrogen-like atom 
in~\cite{Anderson&Kumei&Wulfman1972a,Anderson&Kumei&Wulfman1972b}.
Closely related research on both first- and higher-order symmetry operators of linear Schr\"odinger equations 
and separation variables for such equations was initiated in the same time 
(see \cite{Miller1977} and references therein).

After this ``initial'' stage of research into linear Schr\"odinger equations, 
the study of Lie symmetries was extended to various nonlinear Schr\"odinger equations 
\cite{Doebner&Goldin1994,Fushchich&Moskaliuk1981,Gagnon88a,Gagnon89a,Gagnon89b,Gagnon89c,Gagnon93,Nattermann&Doebner1996,%
Popovych&Ivanova&Eshraghi2003CubicLanl,Popovych&Ivanova&Eshragi2004,Zhdanov&Roman2000}. 
However, the group classification of linear Schr\"odinger equations 
with arbitrary complex-valued potentials still remains an open problem.

Our philosophy is that symmetries underlie physical theories 
and that it is therefore reasonable to look for physically relevant models 
from a set of models (with undetermined parameters) using symmetry criteria.
The selection of possible models is made first
by solving the group classification problem for the (class of) models at hand 
and then choosing a suitable model (or set of models) 
from the list of models obtained in the classification procedure. 
This procedure consists essentially of two parts: 
given a parameterized class of models, 
first determine the symmetry group that is common for all models from the class 
and then describe models admitting symmetry groups 
that are extensions of this common symmetry group~\cite{Ovsiannikov1982}. 

In this paper we carry out the group classification of 
(1+1)-dimensional linear Schr\"odinger equations with complex-valued potentials, having the general form
\begin{gather} \label{LinSchEqs}
i\psi_t+\psi_{xx}+V(t,x)\psi=0,
\end{gather}
where $\psi$ is an unknown complex-valued function  of two real independent variables~$t$ and~$x$ 
and $V$ is an arbitrary smooth complex-valued potential also depending on~$t$ and~$x$.
To achieve this, we apply the algebraic method of group classification (which we describe further on in this paper) 
and reduce the problem of the group classification of the class~\eqref{LinSchEqs} to the classification 
of appropriate subalgebras of the associated equivalence 
algebra~\cite{Bihlo&Cardoso-Bihlo&Popovych2012,Popovych&Kunzinger&Eshragi2010}.
In order to reduce the standard form of Schr\"odinger equations to the form~\eqref{LinSchEqs}, 
we scale $t$ and~$x$ and change the sign of~$V$. 
Note that the larger class of (1+1)-dimensional linear Schr\"odinger equations 
with ``real'' variable mass $m=m(t,x)\ne0$ can be mapped to the class~\eqref{LinSchEqs} 
by a family of point equivalence transformations 
in a way similar to that of gauging  coefficients in linear evolution equations, 
cf.\ \cite{Lie1881withTrans,Ovsiannikov1982,Popovych&Kunzinger&Ivanova2008}. 
Hence the group classification of the class~\eqref{LinSchEqs} also provides 
the group classification of this larger class.

A particular feature of the above equations is that the independent variables $t, x$, on the one hand,
and the dependent variable $\psi$ and arbitrary element $V$, on the other hand, belong to different fields. 
This feature needs a delicate treatment of objects involving~$\psi$ or $V$. 
It is possible to consider Schr\"odinger equations from a ``real perspective'' by representing them  
as systems of two equations for the real and the imaginary parts of~$\psi$, 
but such a representation will only complicate the whole discussion. 
The use of the absolute value and the argument of~$\psi$ instead of the real 
and the imaginary parts is even less convenient since it leads to nonlinear systems 
instead of linear ones. 
This is why we work with complex-valued functions. 
We then need to formally extend the space of variables~$(t,x,\psi)$ with~$\psi^*$ 
and the space of the arbitrary element~$V$ with~$V^*$. 
Here and in what follows star denotes the complex conjugate.
In particular, we consider~$\psi^*$ (resp.\ $V^*$) as an argument for all functions 
depending on~$\psi$ (resp.\ $V$), including components of point transformations and of vector fields.
When we restrict a differential function of~$\psi$ 
to the solution set of an equation from the class~\eqref{LinSchEqs}, 
we also take into account the complex conjugate of the equation, that is 
$-i\psi^*_t+\psi^*_{xx}+V^*(t,x)\psi^*=0$. 
However, it is sufficient to test invariance and equivalence conditions 
only for the original equations since the results of this testing 
will be the same for their complex conjugate counterparts. 
Presenting point transformations, we omit the transformation components for~$\psi^*$ and~$V^*$ 
since they are obtained by conjugating those for~$\psi$~and~$V$.

The structure of this paper is the following: 
In Section~\ref{SectionOnClassesOfDEs} we describe 
the general framework of the group classification of classes of differential equations. 
We define various objects related to point transformations and discuss their properties. 
In Section~\ref{SectionOnUniformlySemi-normalizedClasses} 
we extend the algebraic method of group classification 
to uniformly semi-normalized classes of differential equations. 
We compute the equivalence groupoid, the equivalence group and the equivalence algebra 
of the class~\eqref{LinSchEqs} in Section~\ref{equivgr}. 
It turns out that the class~\eqref{LinSchEqs} has rather good transformational properties: 
it is uniformly semi-normalized with respect to linear superposition of solutions.
In Section~\ref{deteq} 
we then analyze the determining equations for the Lie symmetries of equations from the class~\eqref{LinSchEqs}, 
find the kernel Lie invariance algebra of this class
and single out the classifying condition for admissible Lie symmetry extensions.
In Section~\ref{gclas} we study properties of appropriate subalgebras of the equivalence algebra, 
classify them and complete the group classification of the class~\eqref{LinSchEqs}.
In Section~\ref{LinScEqsalternative} we illustrate 
the advantages of the algebraic method of group classification 
by performing the group classification of the class~\eqref{LinSchEqs} in a different way.
The group classification of (1+1)-dimensional linear Schr\"odinger equations with real potentials 
is presented in Section~\ref{LinSchEqssubclassreal case}.
In the final section we summarize results of the paper.  

\section{Group classification in classes of differential equations}\label{SectionOnClassesOfDEs}

In this section we give the definitions and notation needed 
for the group classification of differential equations. 
For more details see~\cite{Bihlo&Cardoso-Bihlo&Popovych2012,Bihlo&Popovych2017,Opanasenko&Bihlo&Popovych2017,Ovsiannikov1982,Popovych&Kunzinger&Eshragi2010}. 

We begin with a definition of the notion of {\it class} of differential equations.
Let $\mathcal L_\theta$ be a system~$L(x,u_{(p)},\theta_{(q)}(x,u_{(p)}))=0$ of $l$ differential equations $L^1=0$, \dots, $L^l=0$
parameterized by a tuple of arbitrary elements $\theta(x,u_{(p)})=(\theta^1(x,u_{(p)}),
\dots,\theta^k(x,u_{(p)}))$, where $x=(x_1,\dots,x_n)$ is the tuple of independent variables 
and $u_{(p)}$ is the set of the dependent variables $u=(u^1,\dots,u^m)$ together with 
all derivatives of $u$ with respect to $x$ of order less than or equal to $p$. 
The symbol $\theta_{(q)}$ stands for the set of partial derivatives of $\theta$ 
of order less than or equal to $q$ with respect to the variables $x$ and $u_{(p)}$.
The tuple of arbitrary elements $\theta$ runs through the set $\mathcal S$ of solutions 
of an auxiliary system of differential equations $S(x,u_{(p)},\theta_{(q')}(x,u_{(p)}))=0$ and 
differential inequalities $\Sigma(x,u_{(p)},\theta_{(q')}(x,u_{(p)}))\ne 0$ 
(other kinds of inequalities may also appear here), 
in which both $x$ and $u_{(p)}$ play the role of independent variables and $S$ 
and $\Sigma$ are tuples of smooth functions depending on $x,u_{(p)}$ and $\theta_{(q')}$. 
The set $\{\mathcal L_\theta\mid\theta\in\mathcal S\}=:\mathcal L|_{\mathcal S}$ 
is called a {\it class (of systems) of differential equations} 
that is defined by the parameterized form of systems~$\mathcal L_\theta$ 
and the set~$\mathcal S$ run by the arbitrary elements~$\theta$.

Thus, for the class~\eqref{LinSchEqs} we have two partial differential equations (including the complex conjugate equation) 
for two (formally unrelated) dependent variables~$\psi$ and~$\psi^*$ of two independent variables~$t$ and~$x$, 
and two (formally unrelated) arbitrary elements $\theta=(V,V^*)$, which depend only on $t$ and $x$.
Therefore, the auxiliary system for the arbitrary elements of the class~\eqref{LinSchEqs} is
\begin{gather*}
V_\psi=V_{\psi^*}= V_{\psi_t}=V_{\psi^*_t}=V_{\psi_x}=V_{\psi^*_x}=
V_{\psi_{tt}}=V_{\psi^*_{tt}}=V_{\psi_{tx}}=V_{\psi^*_{tx}}=V_{\psi_{xx}}=V_{\psi^*_{xx}}=0,
\\
V^*_\psi=V^*_{\psi^*}= V^*_{\psi_t}=V^*_{\psi^*_t}=V^*_{\psi_x}=V^*_{\psi^*_x}=
V^*_{\psi_{tt}}=V^*_{\psi^*_{tt}}=V^*_{\psi_{tx}}=V^*_{\psi^*_{tx}}=V^*_{\psi_{xx}}=V^*_{\psi^*_{xx}}=0.
\end{gather*}

For a class of differential equations $\mathcal L|_{\mathcal S}$, 
there are objects of various structures that consist of point transformations related to this class.
Let $\mathcal L_\theta$ and $\mathcal L_{\tilde\theta}$ be systems belonging to~$\mathcal L|_{\mathcal S}$. 
We denote by $\mathrm T(\theta,\tilde\theta)$ the set of point transformations in the space of the variables $(x,u)$ 
that map~$\mathcal L_\theta$ to $\mathcal L_{\tilde\theta}$.

An {\it admissible transformation} of the class $\mathcal L|_{\mathcal S}$ is a triple $(\theta,\tilde\theta, \varphi)$
consisting of two arbitrary elements $\theta,\tilde\theta\in \mathcal S$ such that $\mathrm T(\theta,\tilde\theta)\ne\varnothing$
and a point transformation $\varphi \in \mathrm T(\theta,\tilde\theta)$. 
The set of all admissible transformations of the class~$\mathcal L|_{\mathcal S}$, 
\[\mathcal G^\sim=\mathcal G^\sim(\mathcal L|_{\mathcal S}):=\{(\theta,\tilde\theta,\varphi)\mid 
\theta,\tilde\theta\in \mathcal S,\varphi \in \mathrm T(\theta,\tilde\theta)\},\] 
has the structure of a groupoid: 
for any $\theta\in\mathcal S$ the triple $(\theta,\theta,{\rm id})$,  
where ${\rm id}$ is the identity point transformation, is an element of $\mathcal G^\sim$, 
every $(\theta,\tilde\theta, \varphi)\in\mathcal G^\sim$ is invertible 
and $\mathcal G^\sim$ is closed under composition. 
This is why the set~$\mathcal G^\sim$
is called the {\it equivalence groupoid} of the class~$\mathcal L|_{\mathcal S}$. 

The {\it (usual) equivalence (pseudo)group $G^\sim=G^\sim(\mathcal L|_{\mathcal S})$} 
of the class $\mathcal L|_{\mathcal S}$ is defined as being the set of point transformations 
in the joint space of independent and dependent variables, their derivatives and arbitrary elements 
with local coordinates $(x,u_{(p)},\theta)$
that are projectable to the space of $(x,u_{(p')})$ for any $0\leqslant p'\leqslant p$, 
preserve the contact structure on the space with local coordinates $(x,u_{(p)})$,
and map every system from the class~$\mathcal L|_{\mathcal S}$ to a system from the same class. 
Elements of the group~$G^\sim$ are called {\it equivalence transformations}. 
This definition includes two fundamental conditions for general equivalence transformations: 
the preservation of the class~$\mathcal L|_{\mathcal S}$ 
and the preservation of the contact structure on the space with local coordinates $(x,u_{(p)})$. 
The conditions of projectability and locality with respect to arbitrary elements can be weakened, 
and this leads to various generalizations of the notion of equivalence group (see~\cite{Popovych&Kunzinger&Eshragi2010}). 
Note that each equivalence transformation~$\mathcal T\in G^\sim$ generates 
a family of admissible transformations from~$\mathcal G^\sim$, 
$G^\sim\ni\mathcal T\rightarrow\{(\theta,\mathcal T\theta,\mathcal T|_{(x,u)})\mid \theta\in\mathcal S\}\subset\mathcal G^\sim$, 
where $\mathcal T|_{(x,u)}$ is the restriction of~$\mathcal T$ to the space of~$(x,u)$. 
For a generalized equivalence group, 
the restriction of~$\mathcal T$ is made after fixing a value of~$\theta$, 
which can be denoted as $\mathcal T^\theta|_{(x,u)}$.
\looseness=-1

The equivalence group of a subclass~$\mathcal L|_{\mathcal S'}$, $\mathcal S'\subset\mathcal S$, 
of the class~$\mathcal L|_{\mathcal S}$ is called 
a \emph{conditional equivalence group} of the class~$\mathcal L|_{\mathcal S}$ 
that is associated with the subclass~$\mathcal L|_{\mathcal S'}$. 
A useful way of describing the equivalence groupoid~$\mathcal G^\sim$ 
is to classify maximal conditional equivalence groups of the class~$\mathcal L|_{\mathcal S}$ 
up to $G^\sim$-equivalence 
and then to classify (up to an appropriate conditional equivalence) the admissible transformations 
that are not generated by conditional equivalence transformations (see~\cite{Popovych&Kunzinger&Eshragi2010} for more details). 
 
The {\it equivalence algebra} $\mathfrak g^\sim=\mathfrak g^\sim(\mathcal L|_{\mathcal S})$ 
of the class $\mathcal L|_{\mathcal S}$ is defined as the set of generators of one-parameter groups 
of equivalence transformations of the class~$\mathcal L|_{\mathcal S}$.
These generators are vector fields in the space of $(x,u_{(p)},\theta)$, 
that are projectable to the space of $(x,u_{(p')})$ for any $0\leqslant p'\leqslant p$ 
and whose projections to the space of $(x,u_{(p)})$ are the $p$th order prolongations 
of the corresponding projections to the space of $(x,u)$.

The {\it maximal point symmetry (pseudo)group} $G_\theta$ of the system $\mathcal L_\theta$ 
(for a fixed $\theta\in \mathcal S$) is a (pseudo)group of transformations  
that act in the space of independent and dependent variables 
that preserve the solution set of the system $\mathcal L_\theta$. 
Each $G_\theta$ can be interpreted as a vertex group of the equivalence groupoid~$\mathcal G^\sim$.
The intersection $G^\cap=G^\cap(\mathcal L|_{\mathcal S}):=\bigcap_{\theta\in \mathcal S} 
G_\theta$ of all $G_\theta$, $\theta\in \mathcal S$, 
is called the {\it kernel of the maximal point symmetry groups} of systems from~$\mathcal L|_{\mathcal S}$.

The vector fields in the space of $(x,u)$ generating one-parameter subgroups 
of the maximal point symmetry group $G_\theta$ of the system $\mathcal L_\theta$
form a Lie algebra~$\mathfrak g_\theta$ with the Lie bracket defined by commutators of vector fields. 
It is called the {\it maximal Lie invariance algebra} of $\mathcal L_\theta$.
The \emph{kernel invariance algebra} of the class~$\mathcal L|_{\mathcal S}$ is the intersection 
$\mathfrak g^\cap=\mathfrak g^\cap(\mathcal L|_{\mathcal S}):= \bigcap_{\theta\in\mathcal S}\mathfrak g_\theta$ 
of the algebras~$\mathfrak g_\theta$, $\theta\in\mathcal S$.

The classical \emph{group classification problem} for the class~$\mathcal L|_{\mathcal S}$ is 
to list all $G^\sim$-inequivalent values of $\theta\in\mathcal S$ 
for which the corresponding  maximal Lie invariance algebras, $\mathfrak g_\theta$, 
are larger than the kernel invariance algebra~$\mathfrak g^\cap$. 
There may be additional point equivalences between the cases obtained in this way 
and these additional equivalences have then to be incorporated into the results. 
This solves the group classification problem for the class~$\mathcal L|_{\mathcal S}$
up to $\mathcal G^\sim$-equivalence. 

Summing up, objects to be found 
in the course of group classification of the class~$\mathcal L|_{\mathcal S}$ include 
the equivalence groupoid~$\mathcal G^\sim$, 
the equivalence group~$G^\sim$, 
the equivalence algebra~$\mathfrak g^\sim$, 
the kernel invariance algebra~$\mathfrak g^\cap$ and 
a complete list of $G^\sim$-inequivalent (resp.\ $\mathcal G^\sim$-inequivalent) values of $\theta$ 
with the corresponding Lie symmetry extensions of~$\mathfrak g^\cap$. 
Additional point equivalences between classification cases can be computed directly
via looking for pairs of cases with similar Lie invariance algebras
(if two systems are equivalent under an invertible point transformation, 
then their Lie symmetry algebras are isomorphic).
Therefore, the construction of the equivalence groupoid~$\mathcal G^\sim$ can be excluded 
from the procedure of group classification  
if this groupoid is of complicated structure, e.g., 
due to the involved hierarchy of maximal conditional equivalence groups of the class~$\mathcal L|_{\mathcal S}$. 

The classical way of performing a group classification of a class~$\mathcal L|_{\mathcal S}$ 
is to use the \emph{infinitesimal invariance criterion}~\cite{Olver1993,Ovsiannikov1982}: 
under an appropriate nondegeneracy condition for the system $\mathcal L_\theta\in\mathcal L|_{\mathcal S}$,
a vector field $Q=\xi^j(x,u)\p_{x_j}+\eta^a(x,u)\p_{u^a}$ belongs to 
the maximal Lie invariance algebra $\mathfrak g_\theta$ of~$\mathcal L_\theta$ 
if and only if the condition 
\[Q_{(p)}L(x,u_{(p)},\theta_{(q)}(x,u_{(p)}))=0,\]
holds on the manifold $\mathcal L_\theta^p$ defined by the system $\mathcal L_\theta$ 
together with its differential consequences in the jet space $J^{(p)}$. 
Here the indices $j$ and $a$ run from $1$ to $n$ and from $1$ to $m$, respectively, 
and we  use the summation convention for repeated indices. 
$Q_{(p)}$ denotes the standard $p$th prolongation of the vector field~$Q$,
\[
Q_{(p)}=Q+\sum_{0<|\alpha|\leqslant p}\Big(D_1^{\alpha_1}\cdots D_n^{\alpha_n}
\left(\eta^a-\xi^ju^a_j\right)+\xi^ju^a_{\alpha+\delta_j}\Big)\p_{u_{\alpha}^a}.
\]
The tuple $\alpha=(\alpha_1,\dots,\alpha_n)$ is a multiindex, 
$\alpha_j\in \mathbb N\cup 0,|\alpha|:=\alpha_1+\dots+\alpha_n$, 
and $\delta_j$ is the multiindex whose $i$th entry equals~1 and whose other entries are zero. 
The variable~$u^a_{\alpha}$ of the jet space $J^{(p)}$ is identified 
with the derivative $\p^{|\alpha|}u^a/\p x^{\alpha_1}_1\dots\p x^{\alpha_n}_n$.
$D_j=\p_j+u^a_{\alpha+\delta_j}\p_{u^a_{\alpha}}$ is the total derivative operator 
with respect to the variable $x_j$. 

The infinitesimal invariance criterion yields the \emph{system of determining equations} 
for the components of the generators of the one-parameter Lie symmetry groups of systems from the class~$\mathcal L|_{\mathcal S}$, 
where the arbitrary elements~$\theta$ play the role of parameters.
Integrating those determining equations that do not involve arbitrary elements 
gives a preliminary form of the generator components, 
and one must then solve the remaining equations. 
The solution of these remaining equations depends on the values of the arbitrary elements. 
We call these equations the \emph{classifying equations} for the class~$\mathcal L|_{\mathcal S}$. 

In order to find the kernel invariance algebra~$\mathfrak g^\cap$, 
one must first split the determining equations 
with respect to parametric derivatives of arbitrary elements and of dependent variables 
and then solve the system obtained.

Finding Lie symmetry extensions of the kernel Lie algebra depends on an analysis of the classifying equations. 
This part of solving the group classification problem is intricate and 
various techniques are used to obtain the solution. There are two main approaches. 
If the class considered has a simple structure 
(for example, when arbitrary elements are constants or are functions of just one argument), 
then the techniques used rely on the study of the compatibility of the classifying equations 
and their direct solution with respect to both the components of Lie symmetry generators and the arbitrary elements 
(up to the equivalence defined by the equivalence group). 
See, for instance, \cite{Bluman&Kumei,Ovsiannikov1982,Popovych&Ivanova2003a,Vaneeva&Popovych&Sophocleous2009,Vaneeva&Popovych&Sophocleous2012} and the references given there.
For more complicated classes, the direct approach seems to be irrelevant, 
and more advanced algebraic techniques need to be used.

\section{Uniformly semi-normalized classes}\label{SectionOnUniformlySemi-normalizedClasses}

In the most general setting, the main point of the algebraic approach to group classification is 
to classify (up to certain equivalence relation induced by point transformations 
between systems belonging to the class~$\mathcal L|_{\mathcal S}$ under study) 
certain Lie algebras of vector fields related to these systems.%
\footnote{See, e.g., \cite{Basarab&Lahno&Gungor,Basarab&Lahno&Zhadonov,Gagnon93,Gazeau92,Magadeev1993,Popovych&Ivanova&Eshraghi2003CubicLanl,Popovych&Ivanova&Eshragi2004,Zhdanov&Lahno1999} 
and~\cite{Bihlo&Cardoso-Bihlo&Popovych2012,Bihlo&Popovych2017,Opanasenko&Bihlo&Popovych2017,Popovych&Kunzinger&Eshragi2010}
for the application of the preliminary and advanced versions of the algebraic method, respectively, 
to various classes of differential equations.%
}  
The key problem is to select sets of vector fields to be classified 
and the equivalence relation to be used in this classification~\cite{Bihlo&Cardoso-Bihlo&Popovych2012}.
For the application of the algebraic method to be effective, 
the selected objects have to satisfy certain consistency conditions 
which then require particular properties of the equivalence groupoid~$\mathcal G^\sim$ of~$\mathcal L|_{\mathcal S}$. To this end, we begin with some definitions which enable us to formulate our approach. 

We say that the class~$\mathcal L|_{\mathcal S}$ is {\it normalized} 
if its equivalence groupoid~$\mathcal G^\sim$ is generated by its equivalence group~$G^\sim$.
We say that it is {\it semi-normalized}
if the equivalence groupoid~$\mathcal G^\sim$ is generated by 
transformations from~$G^\sim$ and point symmetry transformations of the corresponding source or target systems. 
It is clear that any normalized class of differential equations is semi-normalized.
Normalized classes are especially convenient when one applies the algebraic method of group classification.
If the class~$\mathcal L|_{\mathcal S}$ is normalized, 
then the Lie symmetry extensions of its kernel invariance algebra are obtained via 
the classification of appropriate subalgebras of the equivalence algebra 
whose projections onto the space with local coordinates~$(x,u)$ coincide with 
the maximal Lie invariance algebras of systems from~$\mathcal L|_{\mathcal S}$. 
The property of semi-normalization is useful for determining equivalences between Lie symmetry extensions 
but not for finding such extensions. 
For rigorous definitions and more details, we refer the reader to~\cite{Bihlo&Cardoso-Bihlo&Popovych2012,Popovych&Kunzinger&Eshragi2010}.

Classes of differential equations that are not normalized but have stronger normalization properties 
than semi-normalization often appear in physical applications. 
This is why it is important to weaken the normalization property in such a way 
that still allows us to apply group classification techniques 
analogous to those developed for normalized classes.

\begin{definition}\label{DefinitionOfUniformlySemi-normalizedClasses}
Given a class of differential equations~$\mathcal L|_{\mathcal S}$ 
with equivalence groupoid~$\mathcal G^\sim$ and (usual) equivalence group~$G^\sim$,%
\footnote{A subgroup of the equivalence group can be considered here instead of the entire group.%
} 
suppose that for each $\theta\in\mathcal S$ the point symmetry group~$G_\theta$ 
of the system~$\mathcal L_\theta\in\mathcal L|_{\mathcal S}$ contains a subgroup~$N_\theta$
such that the family $\mathcal N_{\mathcal S}=\{N_\theta\mid\theta\in\mathcal S\}$ of all these subgroups satisfies the following properties:
\begin{enumerate}
\item
$\mathcal T|_{(x,u)}\notin N_\theta$ 
for any $\theta\in\mathcal S$ and any $\mathcal T\in G^\sim$ with $\mathcal T\neq {\rm id}$.
\item
$N_{\mathcal T\theta}=\mathcal T|_{(x,u)} N_\theta(\mathcal T|_{(x,u)})^{-1}$ 
for any $\theta\in\mathcal S$ and any $\mathcal T\in G^\sim$.
\item
For any $(\theta^1,\theta^2,\varphi)\in\mathcal G^\sim$ there exist 
$\varphi^1\in N_{\theta^1}$, $\varphi^2\in N_{\theta^2}$ and $\mathcal T\in G^\sim$ such that 
$\theta^2=\mathcal T\theta^1$ and $\varphi=\varphi^2(\mathcal T|_{(x,u)})\varphi^1$.
\end{enumerate}
Here $\mathcal T|_{(x,u)}$ denotes the restriction of~$\mathcal T$ to the space with local coordinates~$(x,u)$. 
We then say that the class of differential equations~$\mathcal L|_{\mathcal S}$ is \emph{uniformly semi-normalized} 
with respect to the symmetry-subgroup family~$\mathcal N_{\mathcal S}$.
\end{definition}

The qualification ``uniformly'' is justified by the fact that in practically relevant examples of such classes 
all the subgroups~$N_\theta$'s are isomorphic or at least of a similar structure (in particular, of the same dimension). 
The first property in Definition~\ref{DefinitionOfUniformlySemi-normalizedClasses} means that 
the intersection of each subgroup~$N_\theta$ with the restriction of~$G^\sim$ to the space of~$(x,u)$ 
is just the identity transformation.
The second property can be interpreted as
equivariance of equivalence transformations with respect to~$\mathcal N_{\mathcal S}$.
The third property means, essentially, that the entire equivalence groupoid~$\mathcal G^\sim$ 
is generated by equivalence transformations and transformations from uniform point symmetry groups.
One of the symmetry transformations $\varphi^1$ or $\varphi^2$ in the last property may be taken to be the identity. 

Each normalized class of differential equations is uniformly semi-normalized with respect to the trivial family~$\mathcal N_{\mathcal S}$, 
where for each~$\theta$ the group~$N_\theta$ consists of just the identity transformation.
It is also obvious that each uniformly semi-normalized class is semi-normalized.
At the same time, there exist semi-normalized classes that are not uniformly semi-normalized. 
A simple example of such a class is given by the class~${\rm ND}$ of nonlinear diffusion equations of the form $u_t=(f(u)u_x)_x$ with $f_u\ne0$, 
which is a classic example in the group analysis of differential equations~\cite{Ovsiannikov1959,Ovsiannikov1982}.
Such equations with special power nonlinearities of the form $f=c_1(u+c_0)^{-4/3}$ 
have singular symmetry properties within the class~${\rm ND}$. 
This fact does not allow the class~${\rm ND}$ to be normalized, 
although it is semi-normalized. 
The elements from~$\mathcal G^\sim({\rm ND})$ that are not generated by elements of~$G^\sim({\rm ND})$ 
are given by the equivalence transformations of equations with the above power nonlinearities 
that are composed with conformal symmetry transformations of these equations. 
The nonlinear diffusion equation with $f=c_1(u+c_0)^{-4/3}$ 
admits the conformal symmetry group with infinitesimal generator $x^2\p_x-3x(u+c_0)\p_u$, 
but this is not a normal subgroup of the point symmetry group of the equation.

The following result, which we call \emph{the theorem on splitting symmetry groups in uniformly semi-normalized classes}, 
provides a theoretical basis for the algebraic method of group classification of such classes. 

\begin{theorem}\label{TheoremOnSplittingSymGroupsInUniformlySemi-normalizedClasses}
Let a class of differential equations~$\mathcal L|_{\mathcal S}$ be uniformly semi-normalized 
with respect to a symmetry-subgroup family~$\mathcal N_{\mathcal S}=\{N_\theta\mid\theta\in\mathcal S\}$.
Then for each $\theta\in\mathcal S$ the point symmetry group~$G_\theta$ of the system~$\mathcal L_\theta\in\mathcal L|_{\mathcal S}$ 
splits over~$N_\theta$. 
More specifically, $N_\theta$ is a normal subgroup of~$G_\theta$, 
$G^{\rm ess}_\theta=G^\sim|_{(x,u)}\cap G_\theta$ is a subgroup of~$G_\theta$, 
and the group~$G_\theta$ is the semidirect product of~$G^{\rm ess}_\theta$ acting on~$N_\theta$, 
$G_\theta=G^{\rm ess}_\theta\ltimes N_\theta$.  
Here $G^\sim|_{(x,u)}$ denotes the restriction of~$G^\sim$ to the space with local coordinates~$(x,u)$,  
$G^\sim|_{(x,u)}=\{\mathcal T|_{(x,u)}\mid \mathcal T\in G^\sim\}$.
\end{theorem}

\begin{proof}
We fix an arbitrary $\theta\in\mathcal S$ and take an arbitrary~$\varphi\in G_\theta$. 
Then $(\theta,\theta,\varphi)\in\mathcal G^\sim$ and, 
by the third property in Definition~\ref{DefinitionOfUniformlySemi-normalizedClasses}, 
the transformation~$\varphi$ admits the factorization $\varphi=\mathcal T|_{(x,u)}\varphi^1$ 
for some $\mathcal T\in G^\sim$ and some $\varphi^1\in N_\theta$. 
The element~$N_\theta$ of the family~$\mathcal N_{\mathcal S}$ is a subgroup of~$G_\theta$, $N_\theta<G_\theta$ 
and hence the transformation $\varphi^0:=\mathcal T|_{(x,u)}=\varphi(\varphi^1)^{-1}$ 
also belongs to~$G_\theta$, and consequently to $G^\sim|_{(x,u)}\cap G_\theta=:G^{\rm ess}_\theta$, 
which is a subgroup of~$G_\theta$ as it is the intersection of two groups. 
From this, it follows that for any~$\varphi\in G_\theta$ we have the representation $\varphi=\varphi^0\varphi^1$, 
where $\varphi^0\in G^{\rm ess}_\theta$ and $\varphi^1\in N_\theta$.

The first property of Definition~\ref{DefinitionOfUniformlySemi-normalizedClasses} means that 
the intersection of $G^\sim|_{(x,u)}$ and $N_\theta$ consists of just the identity transformation  
so that the intersection $G^{\rm ess}_\theta\cap N_\theta$ contains only the identity transformation. 

For an arbitrary~$\varphi\in G_\theta$ and an arbitrary~$\tilde\varphi\in N_\theta$, 
we consider the composition $\varphi\tilde\varphi\varphi^{-1}$. 
As an element of~$G_\theta$, the transformation~$\varphi$ admits 
the factorization $\varphi=\varphi^0\varphi^1$
with some $\varphi^0\in G^{\rm ess}_\theta$ and some $\varphi^1\in N_\theta$. 
Since $G^{\rm ess}_\theta<G^\sim|_{(x,u)}$, there exists a $\mathcal T\in G^\sim$ 
such that $\mathcal T\theta=\theta$ and $\varphi^0=\mathcal T|_{(x,u)}$. 
By property 2 of Definition~\ref{DefinitionOfUniformlySemi-normalizedClasses},
we obtain $N_\theta=\varphi^0N_\theta(\varphi^0)^{-1}$. 
Hence the composition 
$\varphi\tilde\varphi\varphi^{-1}=
\varphi^0\varphi^1\tilde\varphi(\varphi^1)^{-1}(\varphi^0)^{-1}$ belongs to~$N_\theta$. 
Thus we have that $N_\theta$ is a normal subgroup of~$G_\theta$, $N_\theta\lhd G_\theta$, and so 
$G_\theta=G^{\rm ess}_\theta\ltimes N_\theta$.
\end{proof}

The members of the family~$\mathcal N_{\mathcal S}=\{N_\theta\mid\theta\in\mathcal S\}$ 
are called \emph{uniform point symmetry groups} 
of the systems from the class~$\mathcal L|_{\mathcal S}$, 
and the subgroup~$G^{\rm ess}_\theta$ is called the \emph{essential point symmetry group} of the system~$\mathcal L_\theta$ 
associated with the uniform point symmetry group~$N_\theta$. 
The knowledge of a family of uniform point symmetry groups trivializes them in the following sense: 
since~$G_\theta$ splits over~$N_\theta$ for each~$\theta$,
then we only need to find the subgroup~$G^{\rm ess}_\theta$ in order to construct~$G_\theta$.

The infinitesimal version of Theorem~\ref{TheoremOnSplittingSymGroupsInUniformlySemi-normalizedClasses} 
may be called \emph{the theorem on splitting invariance algebras in uniformly semi-normalized classes}. 
This version follows immediately from Theorem~\ref{TheoremOnSplittingSymGroupsInUniformlySemi-normalizedClasses} 
if we replace the groups with the corresponding algebras of generators of the one-parameter subgroups 
of these groups. 

\begin{theorem}\label{TheoremOnSplittingInvAlgebrasInUniformlySemi-normalizedClasses}
Suppose that a class of differential equations~$\mathcal L|_{\mathcal S}$ is uniformly semi-normalized 
with respect to a family of symmetry subgroups~$\mathcal N_{\mathcal S}=\{N_\theta\mid\theta\in\mathcal S\}$.
Then for each $\theta\in\mathcal S$ the Lie algebras~$\mathfrak g^{\rm ess}_\theta$ and~$\mathfrak n_\theta$ 
that are associated with the groups~$G^{\rm ess}_\theta$ and~$N_\theta$
are, respectively, a subalgebra and an ideal of the maximal Lie invariance algebra~$\mathfrak g_\theta$ 
of the system~$\mathcal L_\theta\in\mathcal L|_{\mathcal S}$.
Moreover, the algebra~$\mathfrak g_\theta$ is the semi-direct sum
$\mathfrak g_\theta=\mathfrak g^{\rm ess}_\theta\lsemioplus\mathfrak n_\theta$, 
and $\mathfrak g^{\rm ess}_\theta=\mathfrak g^\sim|_{(x,u)}\cap \mathfrak g_\theta$, 
where $\mathfrak g^\sim|_{(x,u)}$ denotes the restriction of~$\mathfrak g^\sim$ to the space with local coordinates~$(x,u)$.
\end{theorem}
 
The group classification problem for a uniformly semi-normalized class~$\mathcal L|_{\mathcal S}$ is solved in the following way: 
when computing the equivalence groupoid~$\mathcal G^\sim$ and analyzing its structure, 
we construct a family of uniform point symmetry groups $\mathcal N_{\mathcal S}=\{N_\theta\mid\theta\in\mathcal S\}$, 
which then establishes the uniformly semi-normalization of the class~$\mathcal L|_{\mathcal S}$ 
and yields the corresponding uniform Lie invariance algebras~$\mathfrak n_\theta$'s.
The subgroup $G^{\rm ess}_\theta=G^\sim|_{(x,u)}\cap G_\theta$ 
and the subalgebra~$\mathfrak g^{\rm ess}_\theta=\mathfrak g^\sim|_{(x,u)}\cap \mathfrak g_\theta$ 
which are the complements of $N_\theta$ and~$\mathfrak n_\theta$ respectively, are in general not known on this step. 
By Theorem~\ref{TheoremOnSplittingInvAlgebrasInUniformlySemi-normalizedClasses}, 
we have for each $\theta\in\mathcal S$ that the maximal Lie invariance algebra~$\mathfrak g_\theta$ 
of the system~$\mathcal L_\theta$ is given by the semi-direct sum
$\mathfrak g_\theta=\mathfrak g^{\rm ess}_\theta\lsemioplus\mathfrak n_\theta$.
Essential Lie invariance algebras are subalgebras of $\mathfrak g^\sim|_{(x,u)}$ and 
are mapped onto each other by the differentials of restrictions of equivalence transformations: 
$\mathfrak g^{\rm ess}_{\mathcal T\theta}=(\mathcal T|_{(x,u)})_*\mathfrak g^{\rm ess}_\theta$. 
Consequently, the  group classification of the class~$\mathcal L|_{\mathcal S}$ 
reduces to the classification of appropriate subalgebras of~$\mathfrak g^\sim|_{(x,u)}$
or, equivalently, of the equivalence algebra~$\mathfrak g^\sim$ itself. 

An important case of uniformly semi-normalized classes, which is relevant to the present paper,  
is given by classes of homogeneous linear systems of differential equations. 

Consider a normalized class $\mathcal L^{\rm inh}|_{\mathcal S^{\rm inh}}$ 
of (generally) inhomogeneous linear systems of differential equations $\mathcal L^{\rm inh}_{\theta,\zeta}$'s
of the form $L(x,u_{(p)},\theta_{(q)}(x))=\zeta(x)$, 
where $\theta$ is a tuple of arbitrary elements parameterizing the homogeneous linear left hand side and depending only on~$x$ 
and the right hand side $\zeta$ is a tuple of arbitrary functions of~$x$.
Suppose that the class $\mathcal L^{\rm inh}|_{\mathcal S^{\rm inh}}$ also satisfies the following conditions: 
\begin{itemize}\itemsep=0ex
\item 
Each system from $\mathcal L^{\rm inh}|_{\mathcal S^{\rm inh}}$ is locally solvable. 
\item 
The zero function is the only common solution of the homogeneous systems from $\mathcal L^{\rm inh}|_{\mathcal S^{\rm inh}}$.
\item 
Restrictions of elements of the equivalence group 
$G^\sim_{\rm inh}=G^\sim(\mathcal L^{\rm inh}|_{\mathcal S^{\rm inh}})$ to the space of~$(x,u)$ 
are fibre-preserving transformations whose components for~$u$ are affine in~$u$, 
that is they are of the form $\tilde x_j=X^j(x)$, $\tilde u^a=M^{ab}(x)(u^b+h^b(x))$, 
where \smash{$\det(X^j_{x_{j'}})\ne0$} and $\det(M^{ab})\ne0$.
\end{itemize}
Here the indices $j$ and $j'$ run from $1$ to $n$ and the indices $a$ and $b$ run from $1$ to $m$. 
The functions~$X^{j\,}$'s and~$M^{ab\,}$'s may satisfy additional constraints 
but the $h^{a\,}$'s are arbitrary smooth functions of~$x$.

Any system~$\mathcal L_{\theta\zeta}$ from the class $\mathcal L^{\rm inh}|_{\mathcal S^{\rm inh}}$ is mapped 
to the associated homogeneous system $\mathcal L_{\theta0}$ by the equivalence transformation 
\[
\mathcal T_{\theta\zeta}\colon\quad
\tilde x_j=x_j,\quad 
\tilde u^a=u^a+h^a(x),\quad 
\tilde\theta=\theta,\quad 
\tilde\zeta=\zeta-L(x,h_{(p)}(x),\theta_{(q)}(x)),
\] 
where $h=(h^1,\dots,h^m)$ is a solution of $\mathcal L_{\theta\zeta}$. 
In other words, the class $\mathcal L^{\rm inh}|_{\mathcal S^{\rm inh}}$ is mapped, 
by the family $\{\mathcal T_{\theta\zeta}\}$ of equivalence transformations 
that are nonlocally parameterized by the arbitrary elements~$\theta$'s and~$\zeta$'s 
to the corresponding class $\mathcal L^{\rm hmg}|_{\mathcal S^{\rm hmg}}$ of homogeneous systems.
The transformations from the equivalence group $G^\sim_{\rm inh}$ with $(x,u)$-components 
$\tilde x_j=x_j$, $\tilde u^a=u^a+h^a(x)$, 
where the $h^{a\,}$'s run through the set of smooth functions of~$x$, 
constitute a normal subgroup~$N^\sim_{\rm inh}$ of this group. 
Furthermore, $G^\sim_{\rm inh}$ splits over~$N^\sim_{\rm inh}$ 
since $G^\sim_{\rm inh}=H^\sim_{\rm inh}\ltimes N^\sim_{\rm inh}$, 
where $H^\sim_{\rm inh}$ is the subgroup of~$G^\sim_{\rm inh}$ consisting of those elements with $h^{a\,}=0$. 
The restriction of~$H^\sim_{\rm inh}$ to the space with local coordinates $(x,u,\theta)$ coincides with 
the equivalence group~$G^\sim_{\rm hmg}$ of the class $\mathcal L^{\rm hmg}|_{\mathcal S^{\rm hmg}}$.
The equivalence groupoid~$\mathcal G^\sim_{\rm hmg}$ of this class can be considered 
as the subgroupoid of the equivalence groupoid~$\mathcal G^\sim_{\rm inh}$ 
of the class $\mathcal L^{\rm inh}|_{\mathcal S^{\rm inh}}$
that is singled out by the constraints $\zeta=0$, $\tilde\zeta=0$ for 
the source and target systems of admissible transformations, respectively. 
For each relevant transformational part, the tuple of parameter-functions~$h$ 
is an arbitrary solution of the corresponding source system. 
The systems  $\mathcal L_{\theta\zeta}$ and $\mathcal L_{\tilde\theta\tilde\zeta}$ 
are $G^\sim_{\rm inh}$-equivalent if and only if 
their homogeneous counterparts $\mathcal L_{\theta0}$ and $\mathcal L_{\tilde\theta0}$
are $G^\sim_{\rm hmg}$-equivalent.
Thus the group classification of systems from the class $\mathcal L^{\rm inh}|_{\mathcal S^{\rm inh}}$ 
reduces to the group classification of systems from the class $\mathcal L^{\rm hmg}|_{\mathcal S^{\rm hmg}}$. 

For each~$\theta\in\mathcal S^{\rm hmg}$ 
we denote by~$G^{\rm lin}_{\theta0}$ the subgroup 
of the point symmetry group~$G_{\theta0}$ of~$\mathcal L_{\theta0}$ 
that consists of the linear superposition transformations:
\[\tilde x_j=x_j,\quad \tilde u^a=u^a+h^a(x),\] 
where the tuple~$h$ is a solution of~$\mathcal L_{\theta0}$. 
The family $\mathcal N_{\rm lin}=\{G^{\rm lin}_{\theta0}\mid\theta\in\mathcal S_{\rm hmg}\}$ of all these subgroups
satisfies the properties of Definition~\ref{DefinitionOfUniformlySemi-normalizedClasses}. 
Therefore, the class $\mathcal L^{\rm hmg}|_{\mathcal S^{\rm hmg}}$ is uniformly semi-normalized 
with respect to the family~$\mathcal N_{\rm lin}$. 
We call this kind of semi-normalization, 
which is characteristic for classes of homogeneous linear systems of differential equations, 
\emph{uniform semi-normalization with respect to linear superposition of solutions}.
By Theorem~\ref{TheoremOnSplittingSymGroupsInUniformlySemi-normalizedClasses},
for each~$\theta\in\mathcal S^{\rm hmg}$ the group~$G_{\theta0}$ splits over~$G^{\rm lin}_{\theta0}$, 
and $G_{\theta0}=G^{\rm ess}_{\theta0}\ltimes G^{\rm lin}_{\theta0}$, 
where $G^{\rm ess}_{\theta0}=G^\sim_{\rm hmg}|_{(x,u)}\cap G_{\theta0}$.
By Theorem~\ref{TheoremOnSplittingInvAlgebrasInUniformlySemi-normalizedClasses}, 
the splitting of the point symmetry group induces a splitting of the corresponding maximal Lie invariance algebra: 
$\mathfrak g_{\theta0}=\mathfrak g^{\rm ess}_{\theta0}\lsemioplus\mathfrak g^{\rm lin}_{\theta0}$, 
where $\mathfrak g^{\rm ess}_{\theta0}$ is the essential Lie invariance algebra of~$\mathcal L_{\theta0}$, 
$\mathfrak g^{\rm ess}_{\theta0}=\mathfrak g^\sim_{\rm hmg}|_{(x,u)}\cap\mathfrak g_{\theta0}$, 
and the ideal $\mathfrak g^{\rm lin}_{\theta0}$, being the trivial part of~$\mathfrak g_{\theta0}$, 
consists of vector fields generating one-parameter symmetry groups of linear superposition of solutions. 
Thus, the group classification problem for the class $\mathcal L^{\rm hmg}|_{\mathcal S^{\rm hmg}}$ 
reduces to the classification of appropriate subalgebras 
of the equivalence algebra $\mathfrak g^\sim_{\rm hmg}$ of this class. 
The qualification ``appropriate'' means that the restrictions of these subalgebras to the space of $(x,u)$ 
are essential Lie invariance algebras of systems from $\mathcal L^{\rm hmg}|_{\mathcal S^{\rm hmg}}$.

For a class $\mathcal L^{\rm hmg}|_{\mathcal S^{\rm hmg}}$ 
of linear homogeneous systems of differential equations
that is uniformly semi-normalized with respect to the linear superposition of solutions, 
it is not necessary to start by considering the associated normalized superclass 
$\mathcal L^{\rm inh}|_{\mathcal S^{\rm inh}}$ of generally inhomogeneous linear systems.
The class $\mathcal L^{\rm hmg}|_{\mathcal S^{\rm hmg}}$ itself 
can be the starting point of the analysis. 
In order to get directly its uniform semi-normalization, 
we need to suppose the following properties of~$\mathcal L^{\rm hmg}|_{\mathcal S^{\rm hmg}}$, 
which are counterparts of the above properties of~$\mathcal L^{\rm hmg}|_{\mathcal S^{\rm hmg}}$:
\begin{itemize}
\item 
Each system from $\mathcal L^{\rm hmg}|_{\mathcal S^{\rm hmg}}$ is locally solvable. 
\item 
The zero function is the only common solution of systems from $\mathcal L^{\rm hmg}|_{\mathcal S^{\rm hmg}}$.
\item 
For any admissible transformation $(\theta^1,\theta^2,\varphi)\in\mathcal G^\sim_{\rm hmg}$, 
its transformational part~$\varphi$ is of the form $\tilde x_j=X^j(x)$, $\tilde u^a=M^{ab}(x)(u^b+h^b(x))$, 
where $h=(h^1,\dots,h^m)$ is a solution of $\mathcal L_{\theta^10}$, 
\smash{$\det(X^j_{x_{j'}})\ne0$} and $\det(M^{ab})\ne0$. 
The functions~$X^{j\,}$'s and~$M^{ab\,}$'s may satisfy 
only additional constraints that do not depend on both $\theta^1$ and $\theta^2$. 
\end{itemize}

The class~\eqref{LinSchEqs} of linear Schr\"odinger equations fits very well into the framework 
which we use in the present paper to solve the group classification problem for this class.  

\begin{remark}
There exist classes of homogeneous linear systems of differential equations 
that are uniformly semi-normalized with respect to symmetry-subgroup families 
different from the corresponding families of subgroups of linear superposition transformations. 
See Corollary~\ref{CorollaryOnUniformSemi-nomalizationOf_igammax_subclass} below.
\end{remark}

\begin{remark}
A technique similar to factoring out uniform Lie invariance algebras 
can be applied to kernel invariance algebras in the course of group classification 
of some normalized classes using the algebraic method. 
It is well known that for any class of differential equations~$\mathcal L|_{\mathcal S}$ 
the kernel~$G^\cap$ of the maximal point symmetry groups~$G_\theta$'s 
of systems~$\mathcal L_\theta$'s from~$\mathcal L|_{\mathcal S}$ 
is a normal subgroup of the restriction $G^\sim|_{(x,u)}$ 
of the (usual) equivalence group~$G^\sim$ of~$\mathcal L|_{\mathcal S}$ 
to the space with local coordinates~$(x,u)$, $G^\cap\lhd G^\sim|_{(x,u)}$.  
Analogues of this result can be found in the references
\cite[Proposition~3]{Bihlo&Cardoso-Bihlo&Popovych2011},
\cite[p.~52, Proposition~3.3.9]{Lisle1992}, 
\cite{Ovsiannikov&Ibragimov1975} and \cite[Section~II.6.5]{Ovsiannikov1982}. 
Furthermore, if the class~$\mathcal L|_{\mathcal S}$ is normalized, 
then the kernel~$G^\cap$ is a normal subgroup of each~$G_\theta$, $G^\cap\lhd G_\theta$, $\theta\in \mathcal S$
\cite[Corollary~2]{Bihlo&Cardoso-Bihlo&Popovych2011}. 
However, the groups $G^\sim|_{(x,u)}$ and~$G_\theta$'s do not in general split over~$G^\cap$ 
even if the class~$\mathcal L|_{\mathcal S}$ is normalized. 
Similar assertions for the associated algebras are also true. 
See, in particular, Remark~9 and the subsequent subalgebra classification in~\cite{Bihlo&Cardoso-Bihlo&Popovych2012} 
for a physically relevant example. 
Thus, the splitting of~$G^\sim$ over~$G^\cap$ does not follow from the normalization of~$\mathcal L|_{\mathcal S}$ 
and is an additional requirement 
for the kernel invariance algebra~$\mathfrak g^\cap$ to be factored out 
when the group classification of~$\mathcal L|_{\mathcal S}$ is carried out.
\end{remark}

\section{Equivalence groupoid}\label{equivgr}

In this section we find the equivalence groupoid~$\mathcal G^\sim$ 
and the complete equivalence group~$G^\sim$ of the class~\eqref{LinSchEqs} 
in finite form (that is, not using the infinitesimal method). 
In the following $\mathcal L_V$ denotes the Schr\"odinger equation from the class~\eqref{LinSchEqs} 
with potential $V=V(t,x)$.
We look for all (locally) invertible point transformations of the form 
\begin{equation}\label{equivptrans}
\tilde t=T(t,x,\psi,\psi^*),\quad 
\tilde x =X(t,x,\psi,\psi^*),\quad 
\tilde\psi=\Psi(t,x,\psi,\psi^*),\quad
\tilde\psi^*=\Psi^*(t,x,\psi,\psi^*)
\end{equation}
(that is, $dT\wedge dX\wedge d\Psi\wedge d\Psi^*\ne 0$) 
that map a fixed equation~$\mathcal L_V$ from the class~\eqref{LinSchEqs} 
to an equation~$\mathcal L_{\tilde V}$: $i\tilde \psi_{\tilde t}+\tilde \psi_{\tilde x\tilde x}+\tilde V\tilde \psi=0$ of the same class.

In the rest of this paper, we use the following notation: for any given complex number $\beta$ 
\[
\hat\beta=\beta \quad\mbox{if}\quad T_t>0  \quad\mbox{and}\quad
\hat\beta=\beta^* \quad\mbox{if}\quad T_t<0.
\]

\begin{theorem}\label{thmequivptransf}
The equivalence groupoid~$\mathcal G^\sim$ of the class~\eqref{LinSchEqs} 
consists of triples of the form $(V,\tilde V,\varphi)$, 
where $\varphi$ is a point transformation acting on the space with local coordinates $(t,x,\psi)$ and given by
\begin{subequations}\label{equgroupoid}
\begin{gather}\label{equgroupoid_a}
\tilde t=T, \quad
\tilde x=\varepsilon|T_t|^{1/2}x+X^0,
\\[.5ex]\label{equgroupoid_b}
\tilde \psi=\exp\left(\frac i8\frac{T_{tt}}{|T_t|}\,x^2
 +\frac i2\frac{\varepsilon\varepsilon' X^0_t}{|T_t|^{1/2}}\,x+i\Sigma+\Upsilon\right)
 (\hat\psi+\hat\Phi), 
\end{gather}
$V$ is an arbitrary potential and 
the transformed potential $\tilde V$ is related to~$V$ by the equation
\begin{gather}\label{equgroupoid_c}
\tilde V=\frac{\hat V}{|T_t|}
 +\frac{2T_{ttt}T_t-3T_{tt}^{\,2}}{16\varepsilon'T_t^{\,3}}x^2
 +\frac{\varepsilon\varepsilon'}{2|T_t|^{1/2}}\left(\dfrac{X^0_t}{T_t}\right)_{\!t}x
 -\frac{iT_{tt}+(X^0_t)^2}{4T_t^{\,2}}+\frac{\Sigma_t-i\Upsilon_t}{T_t},
\end{gather}
\end{subequations} 
$T=T(t)$, $X^0=X^0(t)$, $\Sigma=\Sigma(t)$ and $\Upsilon=\Upsilon(t)$ are 
arbitrary smooth real-valued functions of $t$ with $T_t\ne 0$ and
$\Phi=\Phi(t,x)$ denotes an arbitrary solution of the initial equation. 
$\varepsilon=\pm 1$ and $\varepsilon'=\sgn T_t$.
\end{theorem}

\begin{proof}
The class~\eqref{LinSchEqs} is a subclass of the class of generalized Schr\"odinger 
equations of the form $i\psi_t+\psi_{xx}+F=0$,
where $\psi$ is a complex dependent variable of two real independent variables 
$t$ and $x$ and $F=F(t,x,\psi,\psi^*,\psi_x,\psi^*_x)$ is an arbitrary smooth 
complex-valued function of its arguments. 
This superclass is normalized, see~\cite{Popovych&Kunzinger&Eshragi2010}, 
where it was also shown that any admissible transformation of the superclass 
satisfies the conditions
\begin{gather}\label{Condtransf}
T_x=T_\psi=T_{\psi^*}=0, \quad
X_\psi=X_{\psi^*}=0, \quad
X_x^2=|T_t|, \quad  
\Psi_{\hat\psi^*}=0.   
\end{gather}
Hence the same is true for the class~\eqref{LinSchEqs}.
The equations~\eqref{Condtransf} give us 
\[
T=T(t),\quad 
X=\varepsilon|T_t|^{1/2}x+X^0(t),\quad 
\Psi=\Psi(t,x,\hat\psi),
\]
where $T$ and $X^0$ are arbitrary smooth real-valued functions of $t$ 
and $\Psi$ is an arbitrary smooth complex-valued function of its arguments.
Then the invertibility of the transformation gives 
$T_t\ne 0$ and $\smash{\Psi_{\hat\psi}\ne 0}$.
Using the chain rule, we take total derivatives 
of the equation $\tilde\psi(\tilde t,\tilde x)=\Psi(t,x,\hat\psi)$ 
with respect to~$t$ and~$x$, with $\tilde t=T$ and $\tilde x=X$ 
and we find the following expressions for the transformed derivatives: 
\begin{gather*}
\tilde\psi_{\tilde x}=\frac1{X_x}(\Psi_x+\Psi_{\hat\psi}\hat\psi_x),\quad
\tilde\psi_{\tilde t}=\frac1{T_t}(\Psi_t+\Psi_{\hat\psi}\hat\psi_t)-\frac{X_t}{T_tX_x}(\Psi_x
+\Psi_{\hat\psi}\hat\psi_x),\\
\tilde\psi_{\tilde x\tilde x}=\frac1{X_x^{\,2}}(\Psi_{xx}+2\Psi_{x\hat\psi}\hat\psi_x+\Psi_{\hat\psi\hat\psi}\hat\psi_x^{\,2}
+\Psi_{\hat\psi}\hat\psi_{xx}).
\end{gather*}
We substitute these expressions into the equation~$\mathcal L_{\tilde V}$, 
and then take into account that $\psi$ satisfies the equation~\eqref{LinSchEqs} 
so that $\hat\psi_{xx}=-\varepsilon'i\hat\psi_t-\hat V\hat\psi$.  
We then split with respect to~$\hat\psi_t$ and~$\hat\psi_x$, 
which yields  
\begin{gather}\label{equiveqderspla}
\Psi_{\hat\psi\hat\psi}=0,\quad 
\Psi_{x\hat\psi}=\frac i2\frac{X_x}{T_t}X_t\Psi_{\hat\psi},
\\\label{equiveqdersplb}
\varepsilon'i\Psi_t+\Psi_{xx}-\varepsilon'i\frac{X_t}{X_x}\Psi_x
+|T_t|\tilde V\Psi-\hat V\Psi_{\hat\psi}\hat\psi=0.
\end{gather}
The general solution of the first equation in~\eqref{equiveqderspla} 
is $\Psi=\Psi^1(t,x)\hat\psi+\Psi^0(t,x)$, 
where~$\Psi^0$ and~$\Psi^1$ are smooth complex-valued functions of~$t$ and~$x$. 
The second equation in~\eqref{equiveqderspla} then reduces to  
a linear ordinary differential equation with respect to~$\Psi^1$ 
with the independent variable~$x$, and the variable~$t$ plays the role of a parameter. 
Integrating this equation gives the following expression for~$\Psi^1$:  
\[
\Psi^1=\exp\left(\frac i{8}\frac{T_{tt}}{|T_t|}x^2
+\frac i2\frac{\varepsilon\varepsilon'X^0_t}{|T_t|^{1/2}}x+i\Sigma(t)+\Upsilon(t)\right),
\]
where~$\Sigma$ and~$\Upsilon$ are arbitrary smooth real-valued functions of $t$. 
We substitute the expression for~$\Psi$ into the equation~\eqref{equiveqdersplb} 
and then split this equation with respect to~$\hat\psi$ and this then gives us the equation
\[
\tilde V=\frac{\hat V}{|T_t|}-\frac 1{|T_t|}\frac{\Psi^1_{xx}}{\Psi^1}
-\frac i{|T_t|\Psi^1}\left(\Psi^1_t-\frac{X_t}{X_x}\Psi^1_x\right).
\]
which represents the component of the transformation~\eqref{equgroupoid} for~$V$. 
We introduce the function $\Phi=\hat\Psi^0/\hat\Psi^1$, i.e., $\Psi^0=\Psi^1\hat\Phi$. 
The terms in~\eqref{equiveqdersplb} not containing~$\hat\psi$ give an equation in~$\Psi^0$, 
which is equivalent to the initial linear Schr\"odinger equation in terms of~$\Phi$. 
\end{proof}

\begin{corollary}
A (1+1)-dimensional linear Schr\"odinger equation of the form~\eqref{LinSchEqs} 
is equivalent to the free linear Schr\"odinger equation with 
respect to a point transformation if and only if 
$V=\gamma^2(t)x^2+ \gamma^1(t)x+\gamma^0(t)+i\tilde \gamma^0(t)$ 
for some real-valued functions $\gamma^2$, $\gamma^1$, 
$\gamma^0$ and $\tilde \gamma^0$ of $t$.
\end{corollary}

\begin{corollary}\label{eqgpe}
The usual equivalence group~$G^\sim$ of the class~\eqref{LinSchEqs} consists of point 
transformations of the form~\eqref{equgroupoid} with $\Phi=0$.
\end{corollary}

\begin{proof}
Each transformation from~$G^\sim$ generates a family of admissible transformations 
for the class~\eqref{LinSchEqs} and hence is of the form~\eqref{equgroupoid}.
As a usual equivalence group, the group~$G^\sim$ merely consists of point transformations 
in the variables $(t,x,\psi)$ and the arbitrary element~$V$ 
that can be applied to each equation from the class~\eqref{LinSchEqs} 
and whose components for the variables are independent of~$V$. 
The only transformations of the form~\eqref{equgroupoid} that satisfy these requirements are those 
for which $\Phi$ runs through the set of common solutions of all equations 
from the class~\eqref{LinSchEqs}. $\Phi=0$ is the only common solution.
\end{proof}

\begin{remark}\label{RemarkOnProjectibilityOfEquivGroupOf1+1DLinSchEqs}
Consider the natural projection~$\pi$ of the joint space of the variables and the arbitrary element~$V$ 
on the space of the variables only. 
For each transformation~$\mathcal T$ from the equivalence group~$G^\sim$ 
its components for the variables do not depend on~$V$ and are uniquely extended to~$V$. 
These components define a transformation~$\varphi$, 
which can be taken to be the pushforward of~$\mathcal T$ by the projection~$\pi$, 
$\varphi=\pi_*\mathcal T$. 
Therefore, there is a one-to-one correspondence between the equivalence group~$G^\sim$ 
and the group~$\pi_*G^\sim$ consisting of the projected equivalence transformations, 
$G^\sim\ni\mathcal T\mapsto\pi_*\mathcal T\in\pi_*G^\sim$.
\end{remark}

\begin{remark}\label{RemarkOnDiscreteEquivTransOf1+1DLinSchEqs}
The identity component of the equivalence group~$G^\sim$ consists of transformations
of the form~\eqref{equgroupoid} with $T_t>0$ and $\varepsilon=1$. 
This group also contains two discrete transformations that are involutions:
the space reflection $\tilde t=t$, $\tilde x=-x$, $\tilde \psi=\psi$, $\tilde V=V$
and the Wigner time reflection $\tilde t=-t$, $\tilde x=x$, $\tilde \psi=\psi^*$, $\tilde V=V^*$, 
which are independent up to composition with each other and with continuous transformations. 
These continuous and discrete transformations generate the entire equivalence group~$G^\sim$.
\end{remark}

\begin{corollary}\label{corolequivalg}
The equivalence algebra of the class~\eqref{LinSchEqs} is the algebra 
\[
\mathfrak g^\sim=\langle \hat D(\tau),\hat G(\chi),\hat M(\sigma), \hat I(\rho)\rangle, 
\]
where $\tau$, $\chi$, $\sigma$ and~$\rho$ run through the set of smooth real-valued functions of~$t$. 
The vector fields that span~$\mathfrak g^\sim$ are given~by
\begin{gather*}
\hat D(\tau)=\tau\p_t+\frac 1{2}\tau_t x\p_x+\frac i{8}\tau_{tt}x^2(\psi\p_\psi-\psi^*\p_{\psi^*})
\\ \phantom{\hat D(\tau)=}
{}-\left(\tau_t V-\frac 1{8}\tau_{ttt}x^2-i\frac{\tau_{tt}}{4}\right)\p_V
-\left(\tau_t V^*-\frac 1{8}\tau_{ttt}x^2+i\dfrac{\tau_{tt}}{4}\right)\p_{V^*},
\\
\hat G(\chi)=\chi\p_x+\frac i{2}\chi_tx(\psi\p_\psi-\psi^*\p_{\psi^*})+\frac{\chi_{tt}}{2}x(\p_V+\p_{V^*}),
\\
\hat M(\sigma)=i\sigma(\psi\p_\psi-\psi^*\p_{\psi^*})+\sigma_t(\p_V+\p_{V^*}),
\\[.5ex]
\hat I(\rho)=\rho(\psi\p_\psi+\psi^*\p_{\psi^*})-i\rho_t(\p_V+\p_{V^*}).
\end{gather*}
\end{corollary}

\begin{proof}
The equivalence algebra~$\mathfrak g^\sim$ can be computed 
using the infinitesimal Lie method in a way similar to that for finding the Lie invariance algebra
of a single system of differential equations~\cite{Ovsiannikov1982}. 
However, we can avoid these calculations by noting 
that we have already constructed the complete point equivalence group~$G^\sim$. 
The algebra~$\mathfrak g^\sim$ is just the set of infinitesimal generators of one-parameter subgroups 
of the group~$G^\sim$. 
In order to find all such generators, in the transformation form~\eqref{equgroupoid} 
we set~$\Phi=0$ (to single out equivalence transformations), 
$\varepsilon=1$ and $T_t>0$, i.e., $\varepsilon'=1$ 
(to restrict to the continuous component of the identity transformation in~$G^\sim$), 
and represent the parameter function~$\Sigma$ in the form $\Sigma=\frac14X^0X^0_t+\bar\Sigma$ 
with some function~$\bar\Sigma$ of~$t$
(to make the group parameterization more consistent with the one-parameter subgroup structure of~$G^\sim$).
Then we successively take one of the parameter-functions~$T$, $X^0$, $\bar\Sigma$ and~$\Upsilon$ 
to depend on a continuous subgroup parameter~$\delta$, 
set the other parameter-functions to their trivial values, 
which are $t$ for $T$ and zeroes for $X^0$, $\bar\Sigma$ and~$\Upsilon$, 
differentiate the transformation components with respect to~$\delta$ and evaluate the result at $\delta=0$.
The corresponding infinitesimal generator is the vector field 
$\tau\p_t+\xi\p_x+\eta\p_\psi+\eta^*\p_{\psi^*}+\theta\p_V+\theta^*\p_{V^*}$, where 
\[
\tau=\frac{\mathrm d\tilde t}{\mathrm d\delta}\,\bigg|_{\delta=0},\quad
\xi=\frac{\mathrm d\tilde x}{\mathrm d\delta}\,\bigg|_{\delta=0},\quad
\eta=\frac{\mathrm d\tilde\psi}{\mathrm d\delta}\,\bigg|_{\delta=0},\quad
\theta=\frac{\mathrm d\tilde V}{\mathrm d\delta}\,\bigg|_{\delta=0}.
\]
The above procedure gives the vector fields
$\hat D(\tau)$, $\hat G(\chi)$, $\hat M(\sigma)$ and~$\hat I(\rho)$ 
for the parameter-functions~$T$, $X^0$, $\bar\Sigma$ and~$\Upsilon$,
 respectively. 
\end{proof}

Consider the point symmetry group~$G_V$ of an equation~$\mathcal L_V$ from the class~\eqref{LinSchEqs}. 
Each element~$\varphi$ of~$G_V$ generates an admissible point transformation in the class~\eqref{LinSchEqs} 
with the same initial and target arbitrary element~$V$. 
Therefore, the components of~$\varphi$ necessarily have 
the form given in~\eqref{equgroupoid_a}--\eqref{equgroupoid_b}, 
and the parameter-functions satisfy the equation~\eqref{equgroupoid_c} 
with $\tilde V(\tilde t,\tilde x)=V(\tilde t,\tilde x)$. 
The symmetry transformations defined by linear superposition of solutions to the equation~$\mathcal L_V$ 
are of the form given in~\eqref{equgroupoid_a}--\eqref{equgroupoid_b} with $T=t$ and $X^0=\Sigma=\Upsilon=0$.
They constitute a normal subgroup~$G^{\rm lin}_V$ of the group~$G_V$, 
which can be assumed to be the trivial part of~$G_V$.  
The factor group $G_V/G^{\rm lin}_V$ is isomorphic to the subgroup~$G^{\rm ess}_V$ of~$G_V$ 
that is singled out from~$G_V$ by the constraint $\Phi=0$ and will be considered as the only essential part of~$G_V$.
 
\begin{corollary}
The essential point symmetry group~$G^{\rm ess}_V$ of any equation~$\mathcal L_V$ from the class~\eqref{LinSchEqs} 
is contained in the projection~$\pi_*G^\sim$ of the equivalence group~$G^\sim$ 
to the space with local coordinates $(t,x,\psi,\psi^*)$.
\end{corollary}

\begin{corollary}
The class~\eqref{LinSchEqs} is uniformly semi-normalized with respect to linear superposition of solutions.  
\end{corollary}

\begin{proof}
We need to show that any admissible transformation in the class~\eqref{LinSchEqs} 
is the composition of a specific symmetry transformation of the initial equation and 
a transformation from~$G^\sim$. 
We consider two fixed similar equations~$\mathcal L_V$ and~$\mathcal L_{\tilde V}$ in the class~\eqref{LinSchEqs} 
and let $\varphi$ be a point transformation connecting these equations. 
Then $\varphi$ is of the form~\eqref{equgroupoid_a}--\eqref{equgroupoid_b}, 
and the potentials~$V$ and~$\tilde V$ are related by~\eqref{equgroupoid_c}. 
The point transformation $\varphi^1$ given by
$\tilde t=t$, $\tilde x=x$, $\tilde \psi=\psi+\Phi$
with the same function~$\Phi$ as in~$\varphi$ is a symmetry transformation of the initial equation, 
which is related to the linear superposition principle. 
We choose the transformation~$\varphi^2$ to be of the same form~\eqref{equgroupoid_a}--\eqref{equgroupoid_b} 
but with $\Phi=0$. 
By~\eqref{equgroupoid_c}, its extension to the arbitrary element belongs to the group~$G^\sim$. 
The transformation~$\varphi$ is the composition of $\varphi^1$ and~$\varphi^2$.

It is obvious that the transformation~$\varphi$ maps 
the subgroup of linear superposition transformations of the equation~$\mathcal L_V$ onto 
that of the equation~$\mathcal L_{\tilde V}$.
\end{proof}

\section{Analysis of determining equations for Lie symmetries}\label{deteq}

We derive the determining equations for elements from the maximal Lie invariance algebra~$\mathfrak g_V$ 
of an equation~$\mathcal L_V$ from the class~\eqref{LinSchEqs} with potential $V=V(t,x)$.
The general form of a vector field~$Q$ in the space with local coordinates $(t,x,\psi,\psi^*)$ is 
$Q=\tau\p_t+\xi\p_x+\eta\p_\psi+\eta^*\p_{\psi^*}$, 
where the components of~$Q$ are smooth functions of $(t,x,\psi,\psi^*)$. 
The vector field~$Q$ belongs to the algebra~$\mathfrak g_V$ 
if and only if it satisfies the infinitesimal invariance criterion for the equation~$\mathcal L_V$,
which gives
\begin{equation}\label{Deteqcondev}
i\eta^t+\eta^{xx}+\tau V_t\psi+\xi V_x\psi+V\eta=0
\end{equation}
for all solutions of~$\mathcal L_V$. 
Here 
\begin{gather*}
\eta^t=D_t\left(\eta-\tau\psi_t-\xi\psi_x\right)+\tau\psi_{tt}+\xi\psi_{tx},\\
\eta^{xx}=D^2_x\left(\eta-\tau\psi_t-\xi\psi_x\right)+\tau\psi_{txx}+\xi\psi_{xxx},
\end{gather*}
$D_t$ and $D_x$ denote the total derivative operators 
with respect to $t$ and $x$, respectively. 
After substituting $\psi_{xx}=-i\psi_t-V\psi$ and $\psi_{xx}^*=i\psi_t^*-V^*\psi^*$ 
into~\eqref{Deteqcondev} and splitting with respect to 
the other derivatives of~$\psi$ and~$\psi^*$ that occur, 
we obtain a linear overdetermined system of determining equations for the components of~$Q$,
\begin{gather*}
\tau_\psi=\tau_{\psi^*}=0,\quad \tau_x=0,\quad \xi_\psi=\xi_{\psi^*}=0,\quad
\tau_t=2\xi_x,\quad
\eta_{\psi^*}=\eta_{\psi\psi}=0,\quad
2\eta_{x\psi}=i\xi_t,\\
i\eta_t+\eta_{xx}+\tau V_t\psi+\xi V_x\psi+V\eta-(\eta_\psi-\tau_t)V\psi=0.
\end{gather*}
We solve the determining equations in the first line to obtain 
\begin{gather*}
\tau=\tau(t),\quad \xi=\frac 1{2}\tau_tx+\chi(t),\quad
\eta=\left(\frac i{8}\tau_{tt}x^2+\frac i{2}\chi_tx+\rho(t)+i\sigma(t)\right)\psi+\eta^0(t,x),
\end{gather*}
where $\tau$, $\chi$, $\rho$ and~$\sigma$ are smooth real-valued functions of~$t$, 
and $\eta^0$ is a complex-valued function of~$t$ and~$x$.
Then splitting the last determining equation with respect to~$\psi$, 
we derive two equations: 
\begin{gather}\label{classcondForEta0}
i\eta^0_t+\eta^0_{xx}+\eta^0 V=0,
\\ \label{classcond}
\tau V_t+\left(\frac12\tau_tx+\chi\right)V_x+\tau_tV=\frac18\tau_{ttt}x^2+
\frac12\chi_{tt}x+\sigma_t-i\rho_t-\frac i4\tau_{tt}.
\end{gather}
Although both these equations contain the potential~$V$, 
the first equation just means 
that the parameter-function~$\eta^0$ satisfies the equation~$\mathcal L_V$, 
which does not affect the structure of the algebra~$\mathfrak g_V$ when the potential~$V$ varies. 
This is why we only consider the second equation as the real \emph{classifying condition} 
for Lie symmetry operators of equations from the class~\eqref{LinSchEqs}.

We have thus proved the following result:

\begin{theorem}\label{Determing theorem}
The maximal Lie invariance algebra~$\mathfrak g_V$ of the equation~$\mathcal L_V$ 
from the class~\eqref{LinSchEqs} consists of the vector fields of the form 
$
Q=D(\tau)+G(\chi)+\sigma M +\rho I+Z(\eta^0),
$
where
\begin{gather*}
D(\tau)=\tau\p_t+\frac12\tau_tx\p_x+\frac18\tau_{tt}x^2M,\quad
G(\chi)=\chi\p_x+\frac12\chi_txM,\\
M=i\psi\p_\psi-i\psi^*\p_{\psi^*},\quad 
I=\psi\p_\psi+\psi^*\p_{\psi^*},\quad
Z(\eta^0)=\eta^0\p_\psi+\eta^0{^*}\p_{\psi^*},
\end{gather*}
the parameters $\tau$, $\chi$, $\rho$, $\sigma$ run through the set of real-valued 
smooth functions of $t$ satisfying the classifying condition~\eqref{classcond},
and $\eta^0$ runs through the solution set of the equation~$\mathcal L_V$.
\end{theorem} 

Note that Theorem~\ref{Determing theorem} can be derived from  Theorem~\ref{thmequivptransf} 
using the same technique as in Corollary~\ref{corolequivalg}. 
The algebra~$\mathfrak g_V$ consists of infinitesimal generators of one-parameter subgroups  
of the point symmetry group~$G_V$ of the equation~$\mathcal L_V$. 
In considering one-parameter subgroups of~$G_V$, 
we set $\varepsilon=1$ and $T_t>0$, i.e., $\varepsilon'=1$ 
since one-parameter subgroups are contained in the identity component of~$G_V$. 
We also represent~$\Sigma$ in the form $\Sigma=\frac14X^0X^0_t+\bar\Sigma$.
We let the parameter functions~$T$, $X^0$, $\bar\Sigma$, $\Upsilon$ and~$\Phi$
properly depend on a continuous subgroup parameter~$\delta$. 
Then we differentiate the equations~\eqref{equgroupoid} with respect to~$\delta$ 
and evaluate the result at $\delta=0$.
The corresponding infinitesimal generator 
is the vector field $Q=\tau\p_t+\xi\p_x+\eta\p_\psi+\eta^*\p_{\psi^*}$, where 
\[
\tau=\frac{\mathrm d\tilde t}{\mathrm d\delta}\,\bigg|_{\delta=0},\quad
\xi=\frac{\mathrm d\tilde x}{\mathrm d\delta}\,\bigg|_{\delta=0},\quad
\eta=\frac{\mathrm d\tilde\psi}{\mathrm d\delta}\,\bigg|_{\delta=0},
\]
and hence $Q$ has the form given in Theorem~\ref{Determing theorem}. 
The classifying condition~\eqref{classcond} is derived from the equation~\eqref{equgroupoid_c} 
with $\tilde V(\tilde t,\tilde x)=V(\tilde t,\tilde x)$.

In order to find the kernel invariance algebra $\mathfrak g^\cap$ of the class~\eqref{LinSchEqs}, $\mathfrak g^\cap:=\bigcap_V\mathfrak g_V$, 
we vary the potential~$V$ and then split the equations~\eqref{classcondForEta0} and~\eqref{classcond} with respect to~$V$ and its derivatives. 
This gives us the equations $\tau=\chi=0$, $\eta^0=0$ and $\rho_t=\sigma_t=0$.

\begin{proposition}\label{propLinSchEqkernel}
The kernel invariance algebra of the class~\eqref{LinSchEqs} is the algebra $\mathfrak g^\cap=\langle M,I\rangle$.
\end{proposition}

Consider the linear span of all vector fields from Theorem~\ref{Determing theorem} when~$V$ varies given by 
\[
\mathfrak g_\spanindex:=\langle D(\tau),G(\chi),\sigma M,
\rho I,Z(\zeta)\rangle=\textstyle\sum_V\mathfrak g_V. 
\]
Here and in the following
the parameters $\tau$, $\chi$, $\sigma$ and $\rho$ run through the set of real-valued 
smooth functions of $t$, $\zeta$ runs through the set of complex-valued smooth functions of $(t,x)$
and $\eta^0$ runs through the solution set of the equation~$\mathcal L_V$ when the potential~$V$ is fixed. 
We have $\mathfrak g_\spanindex=\textstyle\sum_V\mathfrak g_V$ 
since each vector field~$Q$ from $\mathfrak g_\spanindex$ with nonvanishing~$\tau$ or~$\chi$ 
or with jointly vanishing~$\tau$, $\chi$, $\sigma$ and $\rho$ necessarily 
belongs to~$\mathfrak g_V$ for some~$V$.
The nonzero commutation relations between vector fields spanning~$\mathfrak g_\spanindex$ are 
\begin{gather*}\label{LinSchEqs(1+1)comrel}
[D(\tau^1),D(\tau^2)]=D(\tau^1\tau^2_t-\tau^2\tau^1_t),\quad
[D(\tau),G(\chi)]=G\left(\tau\chi_t-\frac12\tau_t\chi\right),\\
[ D(\tau),\sigma M]=\tau\sigma_t M,\quad
[D(\tau),\rho I]=\tau\rho_t I,\\
[D(\tau),Z(\zeta)]=Z\left(\tau\zeta_t+\frac12\tau_tx\zeta_x-\frac i8\tau_{tt}x^2\zeta\right),\\
[G(\chi^1),G(\chi^2)]=\left(\chi^1\chi^2_t-\chi^2\chi^1_t\right)M,\quad
[G(\chi),Z(\zeta)]=Z\left(\chi\zeta_x-\frac i2\chi_t x\zeta\right),\\
[\sigma M,Z(\zeta)]=Z(-i\sigma \zeta),\quad
[\rho I,Z(\zeta)]=Z(-\rho\zeta).
\end{gather*}
The commutation relations between elements of~$\mathfrak g_\spanindex$ show 
that~$\mathfrak g_\spanindex$ itself is a Lie algebra. 
It is convenient to represent~$\mathfrak g_\spanindex$ as a semi-direct sum, 
\[
\mathfrak g_\spanindex=\mathfrak g^{\rm ess}_\spanindex\lsemioplus\mathfrak g^{\rm lin}_\spanindex, 
\quad\mbox{where}\quad
\mathfrak g^{\rm ess}_\spanindex:=\langle D(\tau),G(\chi),\sigma M,\rho I\rangle 
\quad\mbox{and}\quad
\mathfrak g^{\rm lin}_\spanindex:=\langle Z(\zeta)\rangle
\]
are a subalgebra and an abelian ideal of~$\mathfrak g_\spanindex$, respectively. 
Note that the kernel invariance algebra $\mathfrak g^\cap$ is 
an ideal of $\mathfrak g^{\rm ess}_\spanindex$ and of $\mathfrak g_\spanindex$. 
The above representation of $\mathfrak g_\spanindex$ induces a similar representation 
for each~$\mathfrak g_V$, 
\[
\mathfrak g_V=\mathfrak g^{\rm ess}_V\lsemioplus\mathfrak g^{\rm lin}_V, 
\quad\mbox{where}\quad
\mathfrak g^{\rm ess}_V:=\mathfrak g_V\cap\mathfrak g^{\rm ess}_\spanindex
\quad\mbox{and}\quad
\mathfrak g^{\rm lin}_V:=\mathfrak g_V\cap\mathfrak g^{\rm lin}_\spanindex=\langle Z(\eta^0),\,\eta^0\in\mathcal L_V\rangle
\]
are a finite-dimensional subalgebra (see Lemma~\ref{dimgess} below) and an infinite-dimensional abelian ideal of~$\mathfrak g_V$, respectively. 
We call $\mathfrak g^{\rm ess}_V$ the {\it essential Lie invariance algebra} of the equation~$\mathcal L_V$ for each~$V$. 
The ideal~$\mathfrak g^{\rm lin}_V$ consists of 
vector fields associated with transformations of {\it linear superposition} 
and therefore it is a trivial part of~$\mathfrak g_V$. 

\begin{definition}
A subalgebra~$\mathfrak s$ of~$\mathfrak g^{\rm ess}_\spanindex$ is called \emph{appropriate} 
if there exists a potential~$V$ such that $\mathfrak s=\mathfrak g^{\rm ess}_V$.
\end{definition}

The algebras~$\mathfrak g^{\rm ess}_\spanindex$ and $\mathfrak g^\sim$ are 
related to each other 
by $\mathfrak g^{\rm ess}_\spanindex=\pi_*\mathfrak g^\sim$, 
where $\pi$ is the projection of the joint space of the variables and the arbitrary 
element on the space of the variables only. 
The mapping~$\pi_*$ induced by~$\pi$ is well defined on~$\mathfrak g^\sim$ due to 
the structure of elements of~$\mathfrak g^\sim$.
Note that the vector fields~$\hat D(\tau)$, $\hat G(\chi)$, $\hat M(\sigma)$, $\hat I(\rho)$ 
spanning~$\mathfrak g^\sim$ 
are mapped by~$\pi_*$ to the vector fields~$D(\tau)$, $G(\chi)$, $\sigma M$, $\rho I$ 
spanning~$\mathfrak g^{\rm ess}_\spanindex$, respectively. 
The above relation is stronger than that implied by the specific semi-normalization of 
the class~\eqref{LinSchEqs},  
$\mathfrak g^{\rm ess}_\spanindex\subseteq\pi_*\mathfrak g^\sim$.
Since the algebra~$\mathfrak g^{\rm ess}_\spanindex$ coincides with 
the set~$\pi_*\mathfrak g^\sim$ of infinitesimal generators of one-parameter subgroups 
of the group~$\pi_*G^\sim$, the structure of \smash{$\mathfrak g^{\rm ess}_\spanindex$} is compatible with 
the action of~$\pi_*G^\sim$ on this algebra. 
Moreover, both $\mathfrak g^{\rm ess}_\spanindex$ and $\mathfrak g^{\rm lin}_\spanindex$ 
are invariant with respect to the action of the group~$\pi_*G^\sim$.
This is why the action of~$G^\sim$ on equations from the class~\eqref{LinSchEqs} 
induces the well-defined action of~$\pi_*G^\sim$ on 
the essential Lie invariance algebras of these equations, 
which are subalgebras of $\mathfrak g^{\rm ess}_\spanindex$.  
The kernel $\mathfrak g^\cap$ is obviously an ideal in~$\mathfrak g^{\rm ess}_V$ for any~$V$.

Collecting all the above arguments, we obtain the following assertion. 

\begin{proposition}
The problem of group classification of (1+1)-dimensional linear Schr\"odinger equations 
reduces to the classification of appropriate subalgebras of 
the algebra~$\mathfrak g^{\rm ess}_\spanindex$ 
with respect to the equivalence relation generated by the action of~$\pi_*G^\sim$.
\end{proposition}

Equivalently, we can classify the counterparts of appropriate subalgebras 
in~$\mathfrak g^\sim$ up to~$G^\sim$-equivalence 
and then project them to the space of variables~\cite{Bihlo&Cardoso-Bihlo&Popovych2012,Cardoso-Bihlo&Bihlo&Popovych2011}.

\section{Group classification}\label{gclas}

To classify appropriate subalgebras of the algebra~$\mathfrak g^{\rm ess}_\spanindex$, 
we need to compute the action of transformations from the group~$\pi_*G^\sim$ 
on vector fields from~$\mathfrak g^{\rm ess}_\spanindex$. 
For any transformation $\varphi\in \pi_*G^\sim $ and 
any vector field~$Q\in\mathfrak g^{\rm ess}_\spanindex$,
the pushforward action of~$\varphi$ on~$Q$ is given by
\[
\tilde Q:=\varphi_*Q=Q(T)\p_{\tilde t}+Q(X)\p_{\tilde x}+Q(\Psi)\p_{\tilde \psi}+Q(\Psi^*)\p_{\tilde\psi^*},
\]
where in each component of~$\tilde Q$ we substitute 
the expressions of the variables without tildes in terms of the ``tilded'' variables, 
$(t,x,\psi,\psi^*)=\varphi^{-1}(\tilde t,\tilde x,\tilde \psi,\tilde \psi^*)$, 
and $\varphi^{-1}$ denotes the inverse of~$\varphi$.

For convenience, we introduce the following notation for elementary transformations from $\pi_*G^\sim$, 
which generate the entire group~$\pi_*G^\sim$: 
$\mathcal D(T)$, $\mathcal G(X^0)$, $\mathcal M(\Sigma)$ and $\mathcal I(\Upsilon)$ 
respectively denote the transformations of the form~\eqref{equgroupoid_a}--\eqref{equgroupoid_b} with $\Phi=0$ and $\varepsilon=1$, 
where the parameter-functions~$T$, $X^0$, $\Sigma$ and~$\Upsilon$, successively excluding one of them, 
are set to the values corresponding to the identity transformation, 
which are $t$ for $T$ and zeroes for $X^0$, $\Sigma$ and~$\Upsilon$.
The nontrivial pushforward actions of elementary transformations from~$\pi_*G^\sim$ 
on the vector fields spanning $\mathfrak g^{\rm ess}_\spanindex$~are
\begin{gather*}
\mathcal D_*(T)D(\tau)=\tilde D(T_t\tau),\quad
\mathcal D_*(T)G(\chi)=\tilde G(T^{1/2}_t\chi),\\
\mathcal D_*(T)(\sigma M)=\sigma\tilde M,\quad
\mathcal D_*(T)(\rho I)=\rho\tilde I,\\
\mathcal G_*(X^0)D(\tau)=\tilde D(\tau)+\tilde G\left(\tau X^0_t-\frac12\tau_tX^0\right)+\left(\frac18\tau_{tt}(X^0)^2-\frac14\tau_tX^0X^0_t-\frac12\tau X^0X^0_{tt}\right)\tilde M,\\
\mathcal G_*(X^0)G(\chi)=\tilde G(\chi)+\frac12(\chi X^0_t-\chi_tX^0)\tilde M,\\
\mathcal M_*(\Sigma)D(\tau)=\tilde D(\tau)+\tau\Sigma_t\tilde M,\quad
\mathcal I_*(\Upsilon)D(\tau)=\tilde D(\tau)+\tau\Upsilon_t\tilde I,
\end{gather*}
where in each pushforward by~$\mathcal D_*(T)$ we should substitute the expression for~$t$ given by inverting the relation $\tilde t=T(t)$;
$t=\tilde t$ for the other pushforwards. 
Tildes over vector fields mean that these vector fields are represented in the new variables.

\begin{lemma}\label{dimgess}
$\dim\mathfrak g^{\rm ess}_V\leqslant 7$ for any potential $V$.
\end{lemma}

\begin{proof}
Since we work within the local framework, we can assume 
that the equation $\mathcal L_V$ is considered on a domain of the form $\Omega_0\times\Omega_1$, 
where $\Omega_0$ and $\Omega_1$ are intervals on the $t$- and $x$-axes, respectively. 
Then we successively evaluate the classifying condition~\eqref{classcond} at three different points 
$x=x_0-\delta$, $x=x_0$ and $x=x_0+\delta$ from~$\Omega_1$ for varying~$t$. This gives
 \begin{gather*}
\frac18\tau_{ttt}(x_0-\delta)^2+\frac12\chi_{tt}(x_0-\delta)-i\rho_t+\sigma_t-\frac i4\tau_{tt}=R_1,\\
\frac18\tau_{ttt}x_0^2+\frac12\chi_{tt}x_0-i\rho_t+\sigma_t-\frac i4\tau_{tt}=R_2,\\
\frac18\tau_{ttt}(x_0+\delta)^2+\frac12\chi_{tt}(x_0+\delta)-i\rho_t+\sigma_t-\frac i4\tau_{tt}=R_3,
\end{gather*}
where the right hand sides~$R_1$, $R_2$ and $R_3$ are the results of substituting the above values of $x$ into
$\tau V_t+\left(\frac12\tau_tx+\chi\right)V_x+\tau_tV$.
Combining the above equations and splitting them into real and imaginary parts, 
we obtain a canonical system of linear ordinary differential equations of the form
\begin{gather*}
 \tau_{ttt}=\dots,\quad
 \chi_{tt}=\dots,\quad
 \rho_t=\dots,\quad
 \sigma_t=\dots
\end{gather*}
to $\tau$, $\chi$, $\rho$ and $\sigma$.
The qualification ``canonical'' means that the system is solved 
with respect to the highest-order derivatives.  
The right hand sides of all its equations are denoted by dots 
since their precise form is not important for our argument. 
It is obvious that the solution set of the above system is 
a linear space and parameterized by seven arbitrary constants. 
\end{proof}
In order to classify appropriate subalgebras of~$\mathfrak g^{\rm ess}_\spanindex$, 
for each subalgebra $\mathfrak s$ of~$\mathfrak g^{\rm ess}_\spanindex$ 
we introduce two integers 
\[
k_1=k_1(\mathfrak s):=\dim\pi^0_*\mathfrak s,\quad
k_2=k_2(\mathfrak s):=\dim\mathfrak s\cap\langle G(\chi),\sigma M,\rho I\rangle-2,
\]
where $\pi^0$ denotes the projection onto the space of the variable $t$ 
and $
\pi^0_*\mathfrak s\subset\pi^0_*\mathfrak g^{\rm ess}_\spanindex=\langle \tau\p_t\rangle$. 
The values of $k_1$ and $k_2$
are invariant under the action of~$\pi_*G^\sim$.

\begin{lemma}\label{gclaslema2}
$\pi^0_*\mathfrak g^{\rm ess}_V$ is a Lie algebra for any potential $V$ and
$k_1=\dim\pi^0_*\mathfrak g^{\rm ess}_V\leqslant 3$.
Further, $\pi^0_*\mathfrak g^{\rm ess}_V\in
\{0,\langle \p_t\rangle,\langle \p_t,t\p_t\rangle,\langle \p_t,t\p_t,t^2\p_t\rangle\}
\bmod \pi^0_*G^\sim$. 
\end{lemma}

\begin{proof}
To prove that $\pi^0_*\mathfrak g^{\rm ess}_V$ is a Lie algebra 
we show that it is a linear subspace and closed under Lie bracket of vector fields.
Given $\tau^j\p_t\in \pi^0_*\mathfrak g^{\rm ess}_V$, $ j=1,2$,
there exist $Q^j\in\mathfrak g^{\rm ess}_V$ such that $\pi^0_*Q^j=\tau^j\p_t$. 
Then for any real constants $c_1$ and $c_2$ 
we have $c_1Q^1+c_2Q^2\in \mathfrak g^{\rm ess}_V$.
Therefore, $c_1\tau^1\p_t+c_2\tau^2\p_t=\pi^0_*(c_1Q^1+c_2Q^2)\in\pi^0_*\mathfrak g^{\rm ess}_V$.
Next, $[\tau^1\p_t,\tau^2\p_t]=(\tau^1\tau^2_t-\tau^2\tau^1_t)\p_t
=\pi^0_*[Q^1,Q^2]\in\pi^0_*\mathfrak g^{\rm ess}_V$. 
Further, $\dim\pi^0_*\mathfrak g^{\rm ess}_V\leqslant\dim\mathfrak g^{\rm ess}_V \leqslant 7$.

Thus $\pi^0_*\mathfrak g^{\rm ess}_V$ is a finite-dimensional subalgebra of 
the Lie algebra $\pi^0_*\mathfrak g^{\rm ess}_\spanindex$ 
of vector fields on the real line.
The group $\pi^0_*G^\sim$ coincides with the entire group of local diffeomorphisms of the real line 
and the rest of the lemma follows from Lie's theorem on 
finite-dimensional Lie algebras of vector fields on the real line.
\end{proof}

\begin{lemma}\label{gclaslema}
If a vector field $Q$ is of the form $Q=G(\chi)+\sigma M+\rho I$ with $\chi\neq 0$, 
then $Q=G(1)+\tilde\rho I\bmod\pi_*G^\sim$ for another function $\tilde\rho$. 
\end{lemma}

\begin{proof}
We successively push forward the vector field $Q$ by the transformations 
$\mathcal G(X^0)$ and $\mathcal D(T)$, where $X^0$ and $T$ are arbitrary fixed solutions 
of the ordinary differential equations $\chi X^0_t-\chi_t X^0=2\sigma$ 
and $T_t=1/\chi^2$, respectively. This leads to a vector field of the same form, 
with $\chi=1$ and $\sigma=0$.
\end{proof}

\begin{lemma}\label{lemmatauzero}
If $G(1)+\rho^1 I \in\mathfrak g^{\rm ess}_V$, then also $G(t)+\rho^2 I\in \mathfrak g^{\rm ess}_V$ with
$\rho^2=\int t\rho^1_t \,{\rm d}t$.
\end{lemma}

\begin{proof}
The fact that $G(1)+\rho^1 I \in\mathfrak g^{\rm ess}_V$ means that the values $\tau=\sigma=0$, 
$\chi=1$ and $\rho=\rho^1$ satisfy the classifying condition~\eqref{classcond} 
with the given potential $V$, which gives $V_x=-i\rho_t$. 
Then $tV_x=-it\rho_t$ implies that the classifying condition~\eqref{classcond} is also satisfied by $\tau=\sigma=0$, 
$\chi=t$ and $\rho^2=\int t\rho^1_t \,{\rm d}t$.
\end{proof}

\begin{lemma}
 $\mathfrak g^{\rm ess}_V\cap\langle \sigma M,\rho I\rangle=\mathfrak g^\cap$ for any potential $V$.
\end{lemma}
\begin{proof}
We need to show that $\mathfrak g^{\rm ess}_V\cap\langle \sigma M,\rho I\rangle \subset\mathfrak g^\cap$ 
and $\mathfrak g^\cap\subset\mathfrak g^{\rm ess}_V\cap\langle \sigma M,\rho I\rangle$.
The first inclusion follows from the classifying condition~\eqref{classcond} for $\tau=\chi=0$, 
which implies $\sigma_t=\rho_t=0$. The second inclusion is obvious 
since the kernel invariance algebra $\mathfrak g^\cap$ is contained in $\mathfrak g^{\rm ess}_V$ 
for any $V$. 
\end{proof}
\begin{lemma}\label{lemcondfork_2}
$k_2=\dim\mathfrak g^{\rm ess}_V\cap\langle G(\chi),\sigma M,\rho I\rangle-2 \in \{0,2\}$ 
for any potential $V$. 
\end{lemma}

\begin{proof}
Denote $\mathfrak a_V:=\mathfrak g^{\rm ess}_V\cap\langle G(\chi),\sigma M,\rho I\rangle$. 

If $\mathfrak a_V\subseteq\langle\sigma M,\rho I\rangle$, 
then $\mathfrak a_V=\mathfrak g^{\rm ess}_V\cap\langle \sigma M,\rho I\rangle=\mathfrak g^\cap$, 
i.e., $k_2=\dim \mathfrak a_V-2=0$.

If $\mathfrak a_V\nsubseteq\langle\sigma M,\rho I\rangle$, then there exists 
$Q^1\in\mathfrak a_V$ such that $Q^1\notin\langle\sigma M,\rho I\rangle$. 
From Lemma \ref{gclaslema}, up to $\pi_*G^\sim$-equivalence 
we may assume that $Q^1$ is locally of the form $Q^1=G(1)+\rho^1I$. 
Then Lemma~\ref{lemmatauzero} implies 
that $Q^2=G(t)+\rho^2 I$ with $\rho^2=\int t\rho^1_t\,{\rm d}t$ belongs to~$\mathfrak a_V$. 
We also have $\mathfrak a_V\supset\mathfrak g^\cap$.
It then follows that $\langle M, I, Q^1,Q^2\rangle\subseteq\mathfrak a_V$ and hence
$\dim \mathfrak a_V\geqslant4$.
On the other hand, as follows from the proof of Lemma~\ref{dimgess} under the constraint $\tau=0$, 
the classifying condition~\eqref{classcond} implies in particular 
a canonical system of linear ordinary differential equations of the form
\begin{gather*}
 \chi_{tt}=\dots,\quad
 \rho_t=\dots,\quad
 \sigma_t=\dots
\end{gather*}
in the parameter-functions $\chi$, $\rho$ and $\sigma$, whose solution space is four-dimensional. 
This means that $\dim\mathfrak a_V\leqslant4$. Therefore, $k_2=\dim\mathfrak a_V-2=2$. 
\end{proof}

Summarizing the above results, any appropriate subalgebra of~$\mathfrak g^{\rm ess}_{\langle \,\rangle}$ is spanned by 
\begin{itemize}
\item the basis vector fields $M$ and $I$ of the kernel $\mathfrak g^\cap$, 
\item $k_1$ vector fields $D(\tau^j)+G(\chi^j)+\sigma^jM+\rho^jI$, where $j=1,\dots,k_1$, $k_1\leqslant3$, and $\tau^1,\dots,\tau^{k_1}$ are linearly independent,
 \item $k_2$ vector fields $G(\chi^l)+\sigma^lM+\rho^lI$ where $l=1,\dots,k_2$, $k_2\in\{0,2\}$  and $\chi^1,\dots,\chi^{k_2}$ are linearly independent.
\end{itemize}

\begin{theorem}\label{theoremresults}
A complete list of $G^\sim$-inequivalent (and, therefore, $\mathcal G^\sim$-inequivalent)
Lie symmetry extensions in the class~\eqref{LinSchEqs} is exhausted by the cases collected in Table~1.
\begin{table}[!ht]
\renewcommand{\arraystretch}{1.8}
\begin{center}
Table 1. Results of classification.
$$
\begin{array}{|c|c|c|c|l|}
\hline
\mbox{no.}&k_1&k_2& V &\hfil\mbox{Basis of }\mathfrak g^{\rm ess}_V\\
\hline
1         & 0 & 0 & V(t,x)      & M,\;I \\
2         & 0 & 2 & i\gamma(t)x & M,\;I,\;G(1)-\left(\int\gamma(t)\,{\rm d}t\right)I,\;G(t)-\left(\int t\gamma(t)\,{\rm d}t\right)I\\
3         & 1 & 0 & V(x)        & M,\;I,\;D(1)\\
4\mbox{a} & 1 & 2 & \frac14x^2+ibx,\,b\in\mathbb R_* & M,\;I,\;D(1),\;G(e^t)-be^tI,\;G(e^{-t})+be^{-t}I\\
4\mbox{b} & 1 & 2 & -\frac14x^2+ibx,\,b\in\mathbb R_* & M,\;I,\,D(1),\;G(\cos t)+b(\sin t)I,\;G(\sin t)-b(\cos t)I\\
4\mbox{c} & 1 & 2 & ibx,\,b\in\mathbb R_* & M,\;I,\;D(1),\;G(1)-btI,\;G(t)-\frac12 bt^2I\\
5         & 3 & 0 & cx^{-2},\,c\in\mathbb C_*& M,\;I,\,D(1),\;D(t),\;D(t^2)-\frac12tI\\
6         & 3 & 2 & 0 & M,\;I,\;D(1),\;D(t),\;D(t^2)- \frac12 tI,\;G(1),\;G(t)\\
\hline
\end{array}
$$
\end{center}
\footnotesize
Lie symmetry extensions given in Table~1 are maximal if the parameters involved satisfy 
the following conditions: In~Case~1, the potential~$V$ does not satisfy an equation of the form~\eqref{classcond}.
In~Case~2, the real-valued function~$\gamma$ of~$t$ is constrained by the condition 
$\gamma\ne c_3|c_2t^2+c_1t+c_0|^{-3/2}$ for any real constants $c_0$, $c_1$, $c_2$ and $c_3$ with 
$c_0$, $c_1$ and $c_2$ not vanishing simultaneously. 
In Case~3, $V\ne b_2x^2+b_1x+b_0+c(x+a)^{-2}$ for any real constants~$a$, $b_2$ 
and for any complex constants~$b_1$, $b_0$ and $c$ with $c\mathop{\rm Im}b_1=0$. 
The real constant~$b$ in Cases~4a--4c and the complex constant~$c$ in Case~5 are nonzero.
Further, $b>0\bmod G^\sim$ in Cases~4a and~4b and $b=1\bmod G^\sim$ in Case~4c.
\end{table}
\end{theorem}

\begin{proof}
We consider possible cases for the various values of $k_1$ and $k_2$.

\medskip\par\noindent
$\boldsymbol{k_1=k_2=0.}$ This is the general case with no extension, 
i.e., $\mathfrak g^{\rm ess}_V=\mathfrak g^\cap$ (Case~1 of Table~1).

\medskip\par\noindent
$\boldsymbol{k_1=0,\,k_2=2.}$ Lemmas \ref{gclaslema} and 
\ref{lemmatauzero} imply that up to $G^\sim$-equivalence 
the algebra $\mathfrak g^{\rm ess}_V$ contains the vector fields $G(1)+\rho^1I$ 
and $G(t)+\rho^2I$, where $\rho^1$ is a smooth real-valued function of $t$ 
and $\rho^2=\int t\rho^1_t \,{\rm d}t$.
Integrating the classifying condition~\eqref{classcond} for these vector fields 
with respect to $V$ gives $V=-i\rho_tx+\alpha(t)+i\beta(t)$, 
and $\alpha=\beta=0\bmod G^\sim$. Denoting $-\rho_t$ by $\gamma$, 
we obtain $V=i\gamma(t)x$, $\rho^1=-\int\gamma\,{\rm d}t$ 
and $\rho^2=-\int t\gamma\,{\rm d}t$,
which leads to Case~2 of Table~1.

Let us describe the values of $\gamma$ for which the 
Lie symmetry extension constructed is maximal. 
We substitute the potential $V=i\gamma(t)x$ 
into the classifying condition~\eqref{classcond} and, after splitting with respect to $x$, 
derive the system 
\[                      
\tau_{ttt}=0,\quad \chi_{tt}=0,\quad \sigma_t=0,\quad \rho_t=\chi\gamma -\frac{\tau_{tt}}{4},
\quad \tau\gamma_t+\frac32\tau_t \gamma=0.
\] 
An additional Lie symmetry extension for such a potential 
may be realized only by vector fields with nonzero values of $\tau$. 
Then the integration of the first and last equations of the above system yields 
\[
\tau=c_2t^2+c_1t+c_0,\quad
\gamma=c_3|\tau|^{-3/2}=c_3|c_2t^2+c_1t+c_0|^{-3/2},
\]
where $c_0$, $c_1$, $c_2$ and $c_3$ are real constants.
Therefore, Case~2 presents a maximal Lie symmetry extension 
if $\gamma\neq c_3|c_2t^2+c_1t+c_0|^{-3/2}$ 
for any real constants $c_0$, $c_1$, $c_2$ and $c_3$, 
where the constants $c_0$, $c_1$ and $c_2$ do not vanish simultaneously. 

\medskip\par\noindent
$\boldsymbol{k_1=1,\,k_2=0.}$
The algebra $\mathfrak g^{\rm ess}_V$ necessarily contains a vector field 
$P^0$ of the form $P^0=D(\tau^0)+G(\chi^0)+\sigma^0M+
\rho^0I$, where all the parameter functions are real-valued functions 
of $t$ with $\tau^0\ne 0$. 
Pushing forward $P^0$ by a transformation from $\pi_*G^\sim$, we can set $\tau^0=1$ and   
$\chi^0=\sigma^0=\rho^0=0$. That is, up to $\pi_*G^\sim$-equivalence
we can assume that $P^0=D(1)$, cf.\ Lemma~\ref{gclaslema2}.
The classifying condition~\eqref{classcond} for the vector 
field $P^0$ gives $V_t=0$, which implies Case~3 of Table~1 
with an arbitrary time-independent potential~$V$.

We now find the condition when the Lie symmetry extension obtained is really maximal. 
The presence of any additional extension means that the algebra 
$\mathfrak g^{\rm ess}_V$ necessarily contains a vector field 
$Q=D(\tau)+G(\chi)+\sigma M+\rho I$  with $\tau_t\ne 0$ or $\chi\ne 0$. 
Substituting $Q$ into the classifying condition~\eqref{classcond} 
and fixing a value of $t$ gives a linear ordinary differential 
equation with respect to $V=V(x)$.
The general solution of any such equation is of the form $V=b_2x^2+b_1x+b_0+c(x+a)^{-2}$, 
where $a$ and $b_2$ are real constants and $b_1$, $b_0$ and $c$ are complex constants.
Moreover, the constant $b_1$ is zero if $\tau_t\ne 0$ and, if $\tau_t= 0$ and 
$\chi\ne 0$, we have $c=0$.
Therefore, the Lie symmetry extension of Case~3 is maximal 
if and only if $V\ne~b_2x^2+b_1x+b_0+c(x+a)^{-2}$ for any real constants $a$, $b_2$ 
and for any complex constant  $b_1$, $b_0$ and $c$
with $c\mathop{\rm Im}b_1=0$. 

\medskip\par\noindent
$\boldsymbol{k_1=1,\,k_2=2.}$
In this case a basis of $\mathfrak g^{\rm ess}_V$ consists of the vector 
fields $M$, $I$, $P^0=D(\tau^0)+G(\chi^0)+\sigma^0M+\rho^0 I$ 
and $Q^p=G(\chi^p)+\sigma^pM+\rho^pI$, where all the parameters are real-valued 
functions of $t$ with $\tau^0\ne 0$ and $\chi^1$ and $\chi^2$ being linearly independent. 
Here and in the following the indices $p$ and $q$ run from $1$ to $2$ and 
we sum over repeated indices. 
The vector field $P^0$ is reduced 
to $D(1)$ up to $\pi_*G^\sim$-equivalence, as in the previous case.
The commutation relations of $\mathfrak g^{\rm ess}_V$
\begin{gather*}
[P^0,Q^p]=G(\chi^p_t)+\sigma^p_tM+\rho^p_tI=a_{pq}Q^q+a_{p3}M+a_{p4}I,\\
[Q^1,Q^2]=(\chi^1\chi^2_t-\chi^2\chi^1_t)M=a_0M, 
\end{gather*}
where $a_{pq}$, $a_{p3}$, $a_{p4}$ and $a_0$ are real constants,
yield 
\begin{gather}\label{LinSchEqscasec}
\chi^p_t=a_{pq}\chi^q,\quad \sigma^p_t=a_{pq}\sigma^q+a_{p3},\quad 
\rho^p_t=a_{pq}\rho^q+a_{p4}, \quad \chi^1\chi^2_t-\chi^2\chi^1_t=a_0. 
\end{gather}
The matrix~$(a_{pq})$ is not zero in view of the linear independence of $\chi^1$ and $\chi^2$.
Moreover, the consistency of the system~\eqref{LinSchEqscasec} implies 
that the trace of~$(a_{pq})$ is zero. 
Using equivalence transformations of time scaling, we can further scale the eigenvalues 
of the matrix $(a_{pq})$ with the same nonzero real values. 
Replacing the vector fields $Q^1$ and $Q^2$ by their independent 
linear combinations leads to a matrix similarity transformation of $(a_{pq})$.
Hence, the matrix $(a_{pq})$ can be assumed to be of one of 
the following real Jordan forms:   
\begin{gather}\label{linSchSimirlty}
\begin{pmatrix}
1 & 0 \\
0 &-1\\
\end{pmatrix},\quad 
\begin{pmatrix}
0 &-1 \\
1 &  0\\
\end{pmatrix},\quad
\begin{pmatrix}
0 & 0 \\
1 & 0 \\
\end{pmatrix}.
\end{gather}
The further consideration of each of these forms consists of a few steps:
We integrate the system of differential equations~\eqref{LinSchEqscasec} 
for the chosen form of~$(a_{pq})$, which gives the components 
of the vector fields~$Q^1$ and~$Q^2$.
From the classifying condition~\eqref{classcond} 
for the basis vector fields of the algebra~$\mathfrak g^{\rm ess}_V$ 
we obtain three independent equations for the potential~$V$, including the equation~$V_t=0$. 
These equations must be solved jointly, and their consistency leads 
to additional constraints for constant parameters involved in~$Q^1$ and~$Q^2$. 
The expressions for both the vector fields~$Q^1$ and~$Q^2$ 
and the potential~$V$ can be simplified by equivalence transformations 
and by changing the basis in the algebra~$\mathfrak g^{\rm ess}_V$ 
we can obtain expressions for $Q^1$ and $Q^2$ as in Cases~4a--4c of Table~1. 

Integrating the system~\eqref{LinSchEqscasec} for the first Jordan form, 
we obtain $\chi^1=b_{01}e^t$, $\sigma^1=b_{11}e^t-a_{13}$,
$\rho^1=b_{12}e^t-a_{14}$, $\chi^2=b_{02}e^{-t}$, $\sigma^2=b_{21}e^{-t}+a_{23}$, 
$\rho^2=b_{22}e^{-t}+a_{24}$, 
where $a_{p3}$, $a_{p4}$, $b_{0p}$ and~$b_{pq}$ are real constants. 
Scaling the vector fields $Q^1$ and $Q^2$ and taking linear combinations of them with $M$ and $I$, 
we can set $b_{0q}=1$ and $a_{p3}=a_{p4}=0$.  
The classifying condition~\eqref{classcond} 
for the vector fields $P^0$, $Q^1$ and $ Q^2$ leads to three independent equations in $V$, 
$V_t=0$, $V_x=\frac12x-ib_{12}+b_{11}$ and $V_x=\frac12x+ib_{22}-b_{21}$.
These equations are consistent only if the constant parameters involved in $Q^1$ and $Q^2$ 
satisfy the constraints $-b_{12}=b_{22}=:b$ and $b_{11}=-b_{21}=:-\hat b$.
Then the potential $V$ is of the form $V=\frac14 (x+2\hat b)^2+ibx+c_1+ic_2$ 
for some real constants $ c_1$ and $c_2$. 
We apply the equivalence transformation~\eqref{equgroupoid} with $T=t$, $\smash{X^0=2\hat b}$, 
$\Sigma=-c_1t$, $\Upsilon=c_2 t$, $\varepsilon=1$ and $\Phi=0$ and 
take a linear combination of the transformed vector field $P^0$ with $M$ and $I$.
This allows us to set $\hat b=c_1=c_2=0$ and finally gives
Case~4a of Table~1.

In the same way, we consider the second Jordan form from \eqref{linSchSimirlty}. 
After integrating the corresponding system~\eqref{LinSchEqscasec}, we obtain  
$\chi^1=b_{01}\cos t-b_{02}\sin t$, $\sigma^1= b_{11}\cos t- b_{12}\sin t-a_{23}$, 
$\rho^1=b_{21}\cos t-b_{22}\sin t-a_{24}$, $\chi^2=b_{01}\sin t+b_{02}\cos t$, 
$\sigma^2=b_{11}\sin t+b_{12}\cos t+a_{13}$ and $\rho^2=b_{21}\sin t+b_{22}\cos t+a_{14}$,
where $b_{0q}$ and $b_{pq}$ are real constants.
Combining the vector fields $Q^1$ and $Q^2$ with each other and with $M$ and $I$, 
we can put $b_{01}=1$, $b_{02}=0$ and $a_{p3}=a_{p4}=0$. 
Substituting the components of $P^0$, $Q^1$ and $Q^2$ that we obtain into 
the classifying condition~\eqref{classcond} 
gives three independent equations in~$V$,  
\begin{gather*}
V_t=0,\\
V_x\cos t=-\frac12x \cos t+i(b_{21}\sin t+b_{22}\cos t)-b_{11}\sin t-b_{12}\cos t,\\
V_x\sin t=-\frac12x \sin t-i(b_{21}\cos t-b_{22}\sin t)+ b_{11}\cos t-b_{12}\sin t,
\end{gather*}
with the consistency condition $ b_{11}= b_{21}=0$. 
We denote $b_{12}=:\hat b$ and $b_{22}=:b$. 
Any solution of the above equations for $V$ can be written as
$V=-\frac14 (x+2\hat b)^2+ibx+c_1+ic_2$ for some real constants $c_1$ and $c_2$.     
By applying the equivalence transformation~\eqref{equgroupoid} 
with $T=t$, $\smash{X^0=2\hat b}$, 
$\Sigma=-c_1t$, $\Upsilon=c_2 t$, $\varepsilon=1$ and $\Phi=0$ and 
taking a linear combination of the transformed vector field $P^0$ with $M$ and $I$, 
we can put $\hat b=c_1=c_2=0$.
This yields Case~4b of Table~1. 

Finally, the general solution of the system~\eqref{LinSchEqscasec} 
for the last Jordan form of the matrix $(a_{pq})$ is 
$\chi^1=b_{01}$, $\sigma^1=a_{13}t+b_{11}$, $\rho^1=a_{14}t+b_{21}$, 
$\chi^2=b_{01}t+b_{02}$, $\sigma^2=\frac12a_{13}t^2+(a_{23}+b_{11})t+b_{12}$, 
$\rho^2=\frac12a_{14}t^2+(a_{24}+b_{21})t+b_{22}$, where $b_{0q}$ 
and $b_{pq}$ are real constants. The constants
$b_{pq}$ and $b_{02}$ can be put equal to zero and $b_{01}=1$ by taking a linear combination 
of the vector fields $Q^1$ and $Q^2$ with each other and with $M$ and $I$. 
Then we successively evaluate the classifying condition~\eqref{classcond} for 
the components of the vector fields $P^0$, $Q^1$ and $Q^2$. 
This gives the following equations for~$V$:
\[V_t=0,\quad tV_x=-ia_{14}t+a_{13}t,\quad
tV_x=-ia_{14}t+a_{13}t-ia_{24}+a_{23},
\] 
which are consistent if and only if $a_{23}=a_{24}=0$. 
Any solution of these equations is of the form $V=ibx+\hat bx+c_1+ic_2$, 
where $c_1$ and $c_2$ are real constants and 
we denote $a_{13}=:\hat b$ and $a_{14}=:-b$. Next we apply 
the equivalence transformation~\eqref{equgroupoid} 
with $T=t$, $X^0=-\hat b t^2$, $\Sigma=\frac13\hat b^2t^3-c_1t$, 
$\Upsilon=\frac13b\hat b t^3+c_2t$, $\varepsilon=1$ and $\Phi=0$
and take a linear combination of the vector fields $P^0$ with $Q^2$, $M$ and~$I$.  
In this way the constants $\hat b$, $c_1$ and $c_2$ are set equal to zero, 
which gives Case~4c of Table~1.

The Lie symmetry extensions presented in Cases~4a--4c of Table~1 
are really maximal if $b\ne 0$. Moreover, up to $G^\sim$-equivalence
we can set $b=1$ in Case~2a and $b>0$ in Cases~2b and~2c; 
cf.\ the proof of Theorem~\ref{theoremresultsigammax}.

\medskip\par\noindent
$\boldsymbol{k_1\geqslant2,\,k_2=0.}$
Lemma~\ref{gclaslema2} implies that up to $\pi_*G^\sim$-equivalence 
the algebra $\mathfrak g^{\rm ess}_V$ 
contains at least two operators $P^0=D(1)$ and $P^1=D(t)+G(\chi^1)+\sigma^1M+\rho^1I$. 
Here we also annihilate the tail of $P^0$
with pushing forward $P^0$ by a transformation from $\pi_*G^\sim$.
As $[P^0,P^1]\in \mathfrak g^{\rm ess}_V$,
we~have 
\[
[P^0,P^1]=D(1)+G(\chi^1_t)+\sigma^1_tM+\rho^1_tI=P^0+a_1M+b_1I
\]
for some constants $a_1$ and $b_1$.
Collecting components in the above equality 
gives the system 
$\chi^1_t=0$, $\sigma^1_t=a_1$, $\rho^1_t=b_1$ with the general solution   
$\chi^1=a_2$, $\sigma^1=a_1t+a_0$ and $\rho^1=b_1t+b_0$, 
where $a_2$, $a_0$ and $ b_0$ are real constants of integration.
Pushing forward $P^0$ and $P^1$ with $\mathcal G_*(2a_2)$, 
$\mathcal M_*(-a_1t)$ and $\mathcal I_*(-b_1t)$    
and taking a linear combination of  $P^0$ and  $P^1$ with $M$ and $I$, 
we find that we can set the constants $a_0$, $a_1$, $a_2$, $b_0$ and $b_1$ to zero. 
Therefore, the basis vector field  $P^1$ reduces to the form $P^1=D(t)$, 
whereas the forms of $P^0$, $M$ and $I$ are preserved. 

The classifying condition~\eqref{classcond} for $P^0=D(1)$ and $P^1=D(t)$ 
gives two independent equations in $V$, $V_t=0$ and $x V_x+2V=0$. 
Integrating these equations gives $V=cx^{-2}$,
where $c$ is a complex constant. If $c=0$, then $k_2> 0$, which  
contradicts the case assumption $k_2=0$. Thus, the constant $c$ is nonzero. 
We find the maximal Lie invariance algebra in this case.
We substitute $V=cx^{-2}$ with $c\ne 0$ into the classifying 
condition~\eqref{classcond} and derive the system of differential equations 
for functions parameterizing vector fields from $\mathfrak g^{\rm ess}_V$,
$\tau_{ttt}=0$, $\chi=0$, $\sigma_t=0$, $\rho_t=-\frac14\tau_{tt}$.
The solution of the above system implies that the algebra 
$\mathfrak g^{\rm ess}_V$ is spanned by $M$, $I$, $P^0$, $P^1$ 
and one more vector field $P^2=D(t^2)-\frac12tI$, which gives Case~5 of Table~1.  

\medskip\par\noindent
$\boldsymbol{k_1\geqslant2,\,k_2=2.}$
In this case, the algebra $\mathfrak g^{\rm ess}_V$ necessarily
contains the vector fields
$M$, $I$, $P^l=D(\bar\tau^l)+G(\bar\chi^l)+\bar\sigma^lM+\bar\rho^l$, $l=0,1$ 
and $Q^p=G(\chi^p)+\sigma^pM+\rho^pI$, 
where all the parameters are real-valued smooth functions of $t$ 
with $\bar \tau^0$ and $\bar \tau^1$ (resp. $\chi^1$ and $\chi^2$) 
being linearly independent. 
Recall that the indices $p$ and $q$ run from $1$ to $2$, 
and we sum over repeated indices.
As in the previous case, up to $\pi_*G^\sim$-equivalence the vector 
fields $P^0$ and $P^1$ reduce to the form $P^0=D(1)$ and 
$P^1=D(t)+G(\bar\chi^1)+\bar\sigma^1M+\bar\rho^1I$. 

Since the algebra $\mathfrak g^{\rm ess}_V$ is closed with respect 
to the Lie bracket of vector fields, we have $[P^l,Q^p]\in\mathfrak g^{\rm ess}_V$, i.e.,
\begin{gather*}
[P^0,Q^p]=G(\chi^p_t)+\sigma^p_tM+\rho^p_tI=a_{pq}Q^q+a_{p3}M+a_{p4}I,\\
[P^1,Q^p]=d_{pq}Q^q+d_{p3}M+d_{p4}I,
\end{gather*}
where $a_{pq}$, $a_{p3}$, $ a_{p4}$, 
$d_{pq}$, $d_{p3}$ and $d_{p4}$ are real constants.
Using the above commutation relations with $P^0$ 
in the same way as in the case $k_1=1,\,k_2=2$,
we derive three inequivalent cases for the vector fields $Q^p$ 
depending on the Jordan forms of the matrix $(a_{pq})$ 
presented in~\eqref{linSchSimirlty}.  
For the first and second Jordan forms, 
the commutators~$[P^1,Q^p]$ do not belong to the linear span 
of $P^0$, $P^1$, $ Q^1$, $ Q^2$, $M$ and~$I$. 
Hence these cases are irrelevant.

For the last Jordan form from $\eqref{linSchSimirlty}$, 
up to $G^\sim$-equivalence 
and up to linear combining of the above vector fields, 
we can further assume that  
$Q^1=G(1)-btI$, $Q^2=G(t)-\frac12 bt^2I$ and $V=ibx$ for some real constant $b$. 
We expand the commutation relation for the vector fields $P^1$ and~$Q^1$:
\begin{gather*}
[P^1,Q^1]=-\frac12G(1)-bte^tI-\bar\chi^1_tM= d_{11}Q^1+d_{12}Q^2+ d_{13}M+d_{14}I,
\end{gather*}
and equating components gives $b=0$, i.e., $V=0$.
Substituting the value $V=0$ into the classifying condition~\eqref{classcond} 
and splitting with respect to $x$
yields the system of differential equations
$\tau_{ttt}=0$, $\chi_{tt}=0$, $\sigma_t=0$, $\rho_t=-\frac14\tau_{tt}$.
The solution of this system  for~$V=0$ shows that 
the algebra $\mathfrak g^{\rm ess}_V$ is spanned 
by the vector fields presented in Case~6 of Table~1.
\end{proof}

\begin{remark} 
It might be convenient to completely describe properties of appropriate subalgebras 
before their classification but often such an approach is not justified. 
Thus, Lemma~\ref{gclaslema2} shows that the invariant $k_1$ is not greater than three, i.e., $k_1\in\{0,1,2,3\}$. 
As we proved in Theorem~\ref{theoremresults}, 
this invariant cannot be equal to two. 
The reason is that any (finite-dimensional) subalgebra~$\mathfrak s$ 
of~$\mathfrak g^{\rm ess}_\spanindex$ 
with $\dim\pi^0_*\mathfrak s=2$ is not appropriate since
the  condition of extension maximality is not satisfied.
Therefore, Lemma~\ref{gclaslema2} could be strengthened by the constraint
$k_1\in\{0,1,3\}$.
At the same time, the proof of the condition $k_1\ne 2$ needs realizing 
the major part of the group classification of the class~\eqref{LinSchEqs}.  
 \end{remark}

\section{Alternative proof}\label{LinScEqsalternative}

Here we present an alternative way of classifying Lie symmetry extensions 
in the class~\eqref{LinSchEqs}, in which the invariant $k_2$ is considered as leading.
The case $k_2=0$, after partitioning into the subcases $k_1=0$, $k_1=1$ 
or $k_1\geqslant 2$, results in the same extensions as presented in Table~1 
for these values of $k_1$ and $k_2$.  

Let us consider the case $k_2=2$ more closely.
Lemmas \ref{gclaslema} and~\ref{lemmatauzero} imply that, up to $G^\sim$-equivalence, 
the algebra $\mathfrak g^{\rm ess}_V$ contains the vector fields $G(1)+\rho^1I$ 
and $G(t)+\rho^2I$, where $\rho^1$ is a smooth real-valued function of $t$ 
and $\rho^2=\int t\rho^1_t \,{\rm d}t$.
Integrating the classifying condition~\eqref{classcond} for these vector fields 
with respect to $V$ gives $V=-i\rho_tx+\alpha(t)+i\beta(t)$, 
and $\alpha=\beta=0\bmod G^\sim$. Denoting $-\rho_t$ by $\gamma$, we obtain $V=i\gamma(t)x$.  
Thus, we carry out the group classification of the subclass of equations 
from the class~\eqref{LinSchEqs} with potentials of this form, i.e.,
\begin{gather}\label{igammax subclass}
i\psi_t+\psi_{xx}+i\gamma(t)x\psi=0, 
\end{gather}
where $\gamma(t)$ is an arbitrary real-valued function of $t$, 
which will be taken as the arbitrary element of the subclass instead of $V$.
For potentials of the above form, the equation~\eqref{equgroupoid_c} splits with respect
to $x$ and gives the system of differential equations 
\[
2T_{ttt}T_t-3T_{tt}^2=0,\quad 
\left(\frac{X^0_t}{T_t}\right)_t=0,\quad 
\Sigma_t=\frac{(X^0_t)^2}{4T_t},\quad 
\Upsilon_t=-\frac{T_{tt}}{4T_t}-\varepsilon\frac{\gamma X^0}{|T_t|^{1/2}},
\]
whose general solution is
\begin{gather*}
T=\frac{a_1t+a_0}{a_3t+a_2},\quad 
X^0=b_1T+b_0,\quad 
\Sigma=\frac{b_1^2}4 T+c_1,\quad
\Upsilon=-\frac14\ln|T_t|-\varepsilon\int\frac{\gamma X^0}{|T_t|^{1/2}}\,{\rm d}t+c_0,
\end{gather*}
where $a_i$, $i=0,\dots,3$, $b_j$ and $c_j$, $j=0,1$, 
are real constants with $a_1a_2-a_0a_3\ne 0$ 
and the integral denotes a fixed primitive function for the integrand. 
Since the constants $a_i$, $i=0,\dots,3$, are defined up to a nonzero constant multiplier 
(and thus only three of the constants are essential), 
we set $a_1a_2-a_0a_3=\sgn T_t:=\varepsilon'=\pm1$.

To single out the equivalence groupoid~$\mathcal G^\sim_{\eqref{igammax subclass}}$ of the subclass~\eqref{igammax subclass} 
from the equivalence groupoid $\mathcal G^\sim$ of the whole class~\eqref{LinSchEqs}, 
we substitute the above values of~$T$, $X^0$, $\Sigma$ and $\Upsilon$ into~\eqref{equgroupoid} 
and obtain the following statement:

\begin{theorem}
\begin{subequations}\label{equgroupoidigamamx subclass}
The equivalence groupoid~$\mathcal G^\sim_{\eqref{igammax subclass}}$ of    
the subclass~\eqref{igammax subclass} 
consists of triples of the form $(\gamma,\tilde\gamma,\varphi)$, 
where $\varphi$ is a point transformation in the space of variables, 
whose components~are
\begin{gather}\label{equgroupoidigamamx subclass_a}
\tilde t=T:=\frac{a_1t+a_0}{a_3t+a_2}, \quad
\tilde x=\varepsilon|T_t|^{1/2}x+b_1T+b_0,
\\\label{equgroupoidigamamx subclass_b}
\tilde \psi=\frac c{|T_t|^{1/4}}\exp\left(
\frac i8\frac{T_{tt}}{|T_t|}x^2+\frac i2\varepsilon b_1|T_t|^{1/2}x-\varepsilon\int\gamma\frac{b_1T+b_0}{|T_t|^{1/2}}\,{\rm d}t+i\frac{b_1^2}4 T
\right)(\hat\psi+\hat\Phi), 
\end{gather}
the transformed parameter $\tilde\gamma$ is given in terms of~$\gamma$ as
\begin{gather}\label{equgroupoidigamamx subclass_c}
\tilde \gamma=\frac{\varepsilon\varepsilon'}{|T_t|^{3/2}}\gamma,
\end{gather}
\end{subequations}
$a_0$, $a_1$, $a_2$, $a_3$, $b_0$ and $b_1$ are arbitrary real constants with $a_1a_2-a_0a_3=:\varepsilon'=\pm 1$, 
$c$ is a nonzero complex constant, $\Phi=\Phi(t,x)$ is an arbitrary solution of the initial equation, 
$\varepsilon=\pm 1$.
\end{theorem}

\begin{corollary}\label{eqgpe_{13}}
The (usual) equivalence group~$G^\sim_{\eqref{igammax subclass}}$ of the subclass~\eqref{igammax subclass} 
consists of point transformations of the form~\eqref{equgroupoidigamamx subclass} with $b_0=b_1=0$ and $\Phi=0$.
\end{corollary}

\begin{proof}
We argue in a similar way to Corollary~\ref{eqgpe}:
since each transformation from~$G^\sim_{\eqref{igammax subclass}}$ 
generates a family of admissible transformations in the subclass~\eqref{igammax subclass}, 
it is necessarily of the form~\eqref{equgroupoidigamamx subclass}.
There is only one common solution for the equations from the subclass~\eqref{igammax subclass}: 
the zero function.
Hence the independence of the transformation components for the variables on the arbitrary element~$\gamma$ 
is equivalent to the conditions $b_0=b_1=0$ and $\Phi=0$. 
\end{proof}

\begin{corollary}
The equivalence algebra of the subclass~\eqref{igammax subclass} is the algebra 
\begin{gather*}
\mathfrak g^\sim_{\eqref{igammax subclass}}=\langle\hat D_{\eqref{igammax subclass}}(1),\hat 
D_{\eqref{igammax subclass}}(t), \hat D_{\eqref{igammax subclass}}(t^2),\hat M(1),\hat I(1)\rangle,
\end{gather*}
where, as in Corollary~\ref{corolequivalg},
$\hat M(1)=i(\psi\p_\psi-\psi^*\p_{\psi^*})$, 
$\hat I(1)=\psi\p_\psi+\psi^*\p_{\psi^*}$, and 
\begin{gather*}
\hat D_{\eqref{igammax subclass}}(\tau)=\tau\p_t+\frac12\tau_t x\p_x
+\frac 18\tau_{tt}x^2\hat M(1)-\frac14\tau_t\hat I(1)-\frac32\tau_t\gamma\p_\gamma.
\end{gather*}
\end{corollary}

The proof is analogous to that of Corollary~\ref{corolequivalg}.

\begin{corollary}\label{CorollaryOnCommonSymGroupOfEqsFromSubclass(13)}
For each $\gamma=\gamma(t)$, the equation~$\mathcal L_V$ with $V=i\gamma x$ admits
the group~$G^{\rm unf}_V$ of point symmetry transformations of the form~\eqref{equgroupoidigamamx subclass_a}--\eqref{equgroupoidigamamx subclass_b} 
with $T=t$ and $\varepsilon=1$. 
\end{corollary}

\begin{proof}
The relation~\eqref{equgroupoidigamamx subclass_c} obviously implies that 
for each fixed value of the arbitrary element~$\gamma$,
transformations of the form~\eqref{equgroupoidigamamx subclass_a}--\eqref{equgroupoidigamamx subclass_b} 
with $T=t$ and $\varepsilon=1$ leave this value invariant. 
Other transformations are point symmetries of~$\mathcal L_V$ with $V=i\gamma x$ only for some values of~$\gamma$.
\end{proof}

\begin{corollary}\label{CorollaryOnUniformSemi-nomalizationOf_igammax_subclass}
The subclass~\eqref{igammax subclass} is uniformly semi-normalized 
with respect to the family of uniform point symmetry groups $\{G^{\rm unf}_V\}$ of equations from this subclass 
and the subgroup~$H$ of~$G^\sim_{\eqref{igammax subclass}}$ singled out by the constraint $c=1$.
\end{corollary}

\begin{proof}
It is obvious that for any~$V$ the intersection of~$\pi_*H$ and~$G^{\rm unf}_V$ consists of the identity transformation only. 
Consider an arbitrary admissible transformation $(\gamma,\tilde\gamma,\varphi)$ in the subclass~\eqref{igammax subclass}, 
which maps the equation~$\mathcal L_V$ with $V=i\gamma(t)x$ 
to the equation~$\mathcal L_{\tilde V}$ with $\tilde V=i\tilde\gamma(\tilde t)\tilde x$. 
Then $\varphi$ is of the form~\eqref{equgroupoidigamamx subclass_a}--\eqref{equgroupoidigamamx subclass_b} and thus
$G^{\rm unf}_{\tilde V}=\varphi G^{\rm unf}_V\varphi^{-1}$. 
We denote the dependence of~$\varphi$ on the transformation parameters 
appearing in~\eqref{equgroupoidigamamx subclass_a}--\eqref{equgroupoidigamamx subclass_b} 
by writing $\varphi=\varphi(T,\varepsilon,b_1,b_0,c,\Phi)$.
It is obvious that $\varphi=\varphi^2\varphi^0\varphi^1$, where 
$\varphi^1=\varphi(t,1,0,0,1,\Phi)\in G^{\rm unf}_V$, 
$\varphi^2=\varphi(t,1,b_1,b_0,c,0)\in G^{\rm unf}_{\tilde V}$, and
the transformation $\varphi^0=\varphi(T,\varepsilon,0,0,1,0)$, 
prolonged to~$\gamma$ according to~\eqref{equgroupoidigamamx subclass_c}, 
belongs to~$H$.
\end{proof}

Applying Theorem~\ref{Determing theorem} to equations from the subclass~\eqref{igammax subclass}, 
we consider the classifying condition~\eqref{classcond} for the associated form of potentials,
$V=i\gamma(t)x$, and split this condition with respect to~$x$. 
As a result, we obtain the following system of differential equations 
for the parameters of Lie symmetry vector fields:  
\begin{gather}\label{systemigammasubclass}
\tau_{ttt}=0,\quad\chi_{tt}=0,\quad\sigma_t=0,\quad\rho_t=-\gamma\chi-\frac14\tau_{tt},
\end{gather}
as well as the classifying condition
\begin{gather}\label{classcondigammax subclass}
\left(\gamma|\tau|^{3/2}\right)_t=0.
\end{gather}
It is then clear that the kernel invariance algebra~$\mathfrak g^\cap_{\eqref{igammax subclass}}$ of 
the subclass~\eqref{igammax subclass} is spanned by the vector fields
$M$ and $I$.

\begin{theorem}\label{Theoremigamma subclass}
The maximal Lie invariance algebra $\mathfrak g_V$ of an equation~$\mathcal L_V$ 
for $V=i\gamma(t)x$ is spanned by the vector fields $D_{\eqref{igammax subclass}}(\tau)$, 
$G(1)+\rho^1I$, $G(t)+\rho^2I$, $M$, $I$, $Z(\eta^0)$, where
\begin{gather*}
D_{\eqref{igammax subclass}}(\tau):=D(\tau)-\frac14\tau_tI=\tau\p_t+\frac12\tau_tx\p_x+
\frac18\tau_{tt}x^2M-\frac14\tau_tI,
\end{gather*}
the parameter $\tau$ runs through the set $\mathfrak P_\gamma$ of quadratic polynomials 
in $t$ that satisfy the classifying condition~\eqref{classcondigammax subclass}, 
$\rho^1=-\int\gamma(t)\,{\rm d}t$, $\rho^2=-\int t\gamma(t)\,{\rm d}t$ 
and $\eta^0$ runs through the solution set of the equation~$\mathcal L_V$.
\end{theorem}

Each equation~$\mathcal L_V$ with $V=i\gamma(t)x$, 
belonging to the subclass~\eqref{igammax subclass}, is invariant with respect 
to the Lie algebra $\mathfrak g^{\rm unf}_V=\langle G(1)+\rho^1I,G(t)+\rho^2I,M,I,Z(\eta^0)\rangle$ 
of the group~$G^{\rm unf}_V$, where $\eta^0$ again runs through the solution set of the equation~$\mathcal L_V$.
Such algebras have a similar structure for all equations from the subclass.
The commutation relations between vector fields from $\mathfrak g_\spanindex$
imply that the essential part $\mathfrak g^{\rm ess}_V$ of $\mathfrak g_V $ admits the representation 
$\mathfrak g^{\rm ess}_V=\mathfrak g^{\rm ext}_V\lsemioplus(\mathfrak g^{\rm unf}_V\cap\mathfrak g^{\rm ess}_\spanindex)$, 
where $\mathfrak g^{\rm ext}_V=\{D_{\eqref{igammax subclass}}(\tau)\mid\tau\in\mathfrak P_\gamma\}$ 
is a subalgebra of $\mathfrak g^{\rm ess}_V$, 
and $\mathfrak g^{\rm unf}_V$ is an ideal of $\mathfrak g^{\rm ess}_V\cap\mathfrak g^{\rm ess}_\spanindex$.
Interpreting the above representation, we can say that the algebra $\mathfrak g^{\rm ess}_V$ 
is obtained by extending the algebra $\mathfrak g^{\rm unf}_V\cap\mathfrak g^{\rm ess}_\spanindex$ with elements of $\mathfrak g^{\rm ext}_V $.

Consider the linear span
\[
\mathfrak g^{\rm ext}_\spanindex:=\textstyle\sum\limits_{V=i\gamma(t)x}\mathfrak g^{\rm ext}_V =
\langle D_{\eqref{igammax subclass}}(1),\,D_{\eqref{igammax subclass}}(t),\,D_{\eqref{igammax subclass}}(t^2)\rangle
\subset\pi_*\mathfrak g^\sim_{\eqref{igammax subclass}},
\]
where $\pi$ is the projection of the joint space of the variables and the arbitrary 
element on the space of the variables only. 
The algebra $\mathfrak g^{\rm ext}_\spanindex$ is isomorphic to the algebra ${\rm sl}(2,\mathbb R)$.
The pushforwards of vector fields from $\mathfrak g^{\rm ext}_\spanindex$
by transformations from the group \smash{$\pi_*G^\sim_{\eqref{igammax subclass}}$} constitute 
the inner automorphism group ${\rm Inn}(\mathfrak g^{\rm ext}_\spanindex)$ 
of the algebra $\mathfrak g^{\rm ext}_\spanindex$.
The action of \smash{$G^\sim_{\eqref{igammax subclass}}$} on equations 
from the subclass~\eqref{igammax subclass} induces the action 
of ${\rm Inn}(\mathfrak g^{\rm ext}_\spanindex)$ 
on the subalgebras of the algebra $\mathfrak g^{\rm ext}_\spanindex$.
Consequently, the classification of possible Lie symmetry extensions 
in the subclass~\eqref{igammax subclass} reduces to the classification of subalgebras 
of the algebra ${\rm sl}(2,\mathbb R)$, which is well known.

\begin{theorem}\label{theoremresultsigammax}
A complete list of $G^\sim_{\eqref{igammax subclass}}$-inequivalent 
(and, therefore, $\mathcal G^\sim_{\eqref{igammax subclass}}$-inequivalent) 
Lie symmetry extensions in the subclass~\eqref{igammax subclass} 
is given by Table~2.
\begin{table}[!ht]
\renewcommand{\arraystretch}{1.8}
\begin{center}
Table 2. Results of the group classification of the subclass~\eqref{igammax subclass}.
$$
\begin{array}{|c|c|c|l|}
\hline
\mbox{no.}&k_1& V &\hfil\mbox{Basis of }\mathfrak g^{\rm ess}_V\\
\hline
1         & 0 & i\gamma(t)x                         & M,\;I,\;G(1)-\left(\int\gamma(t)\,{\rm d}t\right)I,\;G(t)-\left(\int t\gamma(t)\,{\rm d}t\right)I\\
2\mbox{a} & 1 & ibx,\,b\in\mathbb R_*               & M,\;I,\;G(1)-btI,\;G(t)-\frac12 bt^2I,\;D(1)\\
2\mbox{b} & 1 & ib|t|^{-3/2}x,\,b\in\mathbb R_*     & M,\;I,\;G(1)+2bt|t|^{-3/2}I,\;G(t)-2b|t|^{1/2}I,\;D(t)\\
2\mbox{c} & 1 & ib(t^2+1)^{-3/2}x,\,b\in\mathbb R_* & M,\;I,\;G(1)-bt(t^2+1)^{-1/2}I,\;G(t)+b(t^2+1)^{-1/2}I,\\
          &   &                                     &                     D(t^2+1)-\frac12 tI\\
3         & 3 & 0                                   & M,\;I,\;G(1),\;G(t),\;D(1),\;D(t),\;D(t^2)- \frac12 tI\\
\hline
\end{array}
$$
\end{center}
\footnotesize
Lie symmetry extension given in Case~1 of Table~2 is maximal if and only if
the arbitrary element~$\gamma$ is of the form $\gamma\ne c_3|c_2t^2+c_1t+c_0|^{-3/2}$ 
for any real constants $c_0$, $c_1$, $c_2$ and $c_3$ with $c_0$, $c_1$ 
and $c_2$ not vanishing simultaneously. 
$b=1\bmod G^\sim$ in Case~2a and $b>0\bmod G^\sim$ in Cases~2b and~2c.
\end{table}
\end{theorem}

\begin{proof}
An optimal set of subalgebras of the algebra $\mathfrak g^{\rm ext}_\spanindex$ 
is given by
\[
\{0\}, \,\ 
\langle D(1)\rangle, \,\
\langle D(t)\rangle, \,\
\langle D(t^2+1)-\tfrac12tI\rangle, \,\
\langle D(1),D(t)\rangle, \
\langle D(1),D(t),D(t^2+1)-\tfrac12tI\rangle.
\]

The zero subalgebra gives the general case with no extension of $\mathfrak g^{\rm unf}_V$, 
which is Case~1 of Table~2.

For the one-dimensional subalgebras, we substitute the corresponding values of $\tau $, $\tau=1$, $\tau=t$
and $\tau=t^2+1$ into the classifying condition~\eqref{classcondigammax subclass}, 
integrate the resulting equations with respect to $\gamma$ and obtain Cases~2a--2c of Table~2, respectively.
Using equivalence transformations that do not change the form of $\gamma$, 
we can set $b=1$ in Case~2a and $b>0$ in Cases~2b and~2c.

Similarly, the classifying condition~\eqref{classcondigammax subclass} 
for the two-dimensional subalgebra gives an overdetermined system of two equations 
with $\tau=1$ and $\tau=t$, for which the only solution is $\gamma=0$. 
The maximal extension of $\mathfrak g^{\rm unf}_V$ 
for $\gamma=0$ is three-dimensional and is given by the last subalgebra of the list. 
This gives Case~3 of Table~1. 
\end{proof}

All cases presented in Table~2 are related to those of Table~1.
In the symbol T.N, used in the following, 
T denotes the table number and N is the case number (in Table~T). 
Thus, Cases~2.1, 2.2a and~2.3
coincide with Cases~1.2, 1.4c and 1.6, respectively. 
Some cases are connected via equivalence transformations, 
which are of the form~\eqref{equgroupoid},
\begin{gather*}
2.2b\,\to 1.4a:\quad T=\frac{\sgn t}4 \ln |t|,\quad X^0=\Sigma=\Upsilon=0,\quad\Phi=0;\\
2.2c\,\to1.4b:\quad T=\arctan t,\quad X^0=\Sigma=\Upsilon=0,\quad \Phi=0.
\end{gather*}
Thus, the result of group classification of 
the class~\eqref{LinSchEqs} can be reformulated with involving Table~2.

\begin{corollary}\label{mainresultsclassifiaction}
A complete list of inequivalent Lie symmetry extensions in the class~\eqref{LinSchEqs} is exhausted 
by Cases~1, 3 and~5 of Table~1 and the cases collected in Table~2.
\end{corollary}

\section{Subclass with real-valued potentials}\label{LinSchEqssubclassreal case}

We derive results on group analysis of the subclass~${\rm Sch}_{\mathbb R}$ of equations 
of the form~\eqref{LinSchEqs} with real-valued potentials using those for 
the whole class~\eqref{LinSchEqs}. The condition that potentials 
are real valued leads to additional constraints 
for transformations and infinitesimal generators.

\begin{theorem}\label{thmequivptransfsub}
The equivalence groupoid~$\mathcal G^\sim_{\mathbb R}$ of 
the subclass~${\rm Sch}_{\mathbb R}$  
consists of triples of the form $(V,\tilde V,\varphi)$, 
where $\varphi$ is a point transformation in the space of variables, 
whose components are
\begin{subequations}\label{equgroupoidsub}
\begin{gather}\label{equgroupoidsub_a}
\tilde t=T,\quad
\tilde x=\varepsilon|T_t|^{1/2}x+X^0,
\\[.5ex]\label{equgroupoidsub_b}
\tilde \psi= \frac a{|T_t|^{1/4}}\exp\left(\frac i8\frac{T_{tt}}{|T_t|}\,x^2
 +\frac i2\frac{\varepsilon\varepsilon' X^0_t}{|T_t|^{1/2}}\,x+i\Sigma\right)
 (\hat\psi+\hat\Phi), 
\end{gather}
the transformed potential $\tilde V$ is expressed in terms of~$V$ as
\begin{gather}\label{equgroupoidsub_c}
\tilde V=\frac V{|T_t|}
 +\frac{2T_{ttt}T_t-3T_{tt}^{\,2}}{16\varepsilon'T_t^{\,3}}x^2
 +\frac{\varepsilon\varepsilon'}{2|T_t|^{1/2}}\left(\dfrac{X^0_t}{T_t}\right)_{\!t}x
 -\frac{(X^0_t)^2}{4T_t^{\,2}}+\frac{\Sigma_t}{T_t},
\end{gather}
\end{subequations}
 $T=T(t)$, $X^0=X^0(t)$ and $\Sigma=\Sigma(t)$ are 
arbitrary smooth real-valued functions of $t$ with $T_t\ne 0$ and 
$\Phi=\Phi(t,x)$ is an arbitrary solution of the initial equation.  
$a$ is a nonzero real constant, $\varepsilon=\pm 1$ and $\varepsilon'=\sgn T_t$.
\end{theorem}

\begin{corollary}\label{eqgpesub}
The subclass ${\rm Sch}_{\mathbb R}$ is uniformly semi-normalized with respect to linear superposition of solutions. 
Its equivalence group~$G^\sim_{\mathbb R}$ consists of point 
transformations of the form~\eqref{equgroupoidsub} with $\Phi=0$.
\end{corollary}

\begin{corollary}\label{corolequivalgsub}
The equivalence algebra of the subclass ${\rm Sch}_{\mathbb R}$ is the algebra 
\[
\mathfrak g^\sim_{\mathbb R}=\langle \hat D_{\mathbb R}(\tau),\hat G_{\mathbb R}(\chi),
\hat M_{\mathbb R}(\sigma),\hat I_{\mathbb R}\rangle 
\]
where $\tau$, $\chi$ and $\sigma$ run through the set 
of smooth real-valued functions of~$t$. 
The vector fields $\hat D_{\mathbb R}(\tau)$, $\hat G_{\mathbb R}(\chi)$,
 $\hat M_{\mathbb R}(\sigma)$ and $\hat I_{\mathbb R}$ are given by
\begin{gather*}
\hat D_{\mathbb R}(\tau)=\tau\p_t+\frac12\tau_t x\p_x+\frac i8\tau_{tt}x^2(\psi\p_\psi-\psi^*\p_{\psi^*})
-\frac14 \tau_t\hat I_{\mathbb R}-\left(\tau_t V-\frac18\tau_{ttt}x^2\right)\p_V,\\
\hat G_{\mathbb R}(\chi)=\chi\p_x+\frac i2\chi_tx(\psi\p_\psi-\psi^*\p_{\psi^*})+\frac{\chi_{tt}}2x\p_V,\\[1ex]
\hat M_{\mathbb R}(\sigma)=i\sigma(\psi\p_\psi-\psi^*\p_{\psi^*})+\sigma_t\p_V,\quad
\hat I_{\mathbb R}=\psi\p_\psi+\psi^*\p_{\psi^*}.
\end{gather*}
\end{corollary}

\begin{corollary}
A (1+1)-dimensional linear Schr\"odinger equation of the form~\eqref{LinSchEqs} 
with a real-valued potential $V$ is equivalent to the free linear Schr\"odinger equation 
with respect to a point transformation if and only if the potential is a quadratic polynomial in $x$, 
i.e., $V=\gamma^2(t)x^2+ \gamma^1(t)x+\gamma^0(t)$ 
for some smooth real-valued functions $\gamma^0$, $\gamma^1$
and $\gamma^2$ of $t$.
\end{corollary}

A study of the determining equations for Lie symmetries 
of equations from the subclass ${\rm Sch}_{\mathbb R}$ 
shows that the classifying condition in this 
case is of the form~\eqref{classcond} with $\rho_t=-\frac14\tau_{tt}$,
\begin{equation}\label{classcondsub}
\tau V_t+\left(\frac12\tau_tx+\chi\right)V_x+\tau_tV=\frac18\tau_{ttt}x^2+\frac12\chi_{tt}x+\sigma_t.
\end{equation}
The kernel invariance algebra~$\mathfrak g^\cap_{\mathbb R}$ of 
the subclass ${\rm Sch}_{\mathbb R}$ coincides with the kernel invariance 
algebra $\mathfrak g^\cap$ of the whole class~\eqref{LinSchEqs}, 
cf. Proposition~1.

\begin{theorem}
The maximal Lie invariance algebra $\mathfrak g_V$ of an equation~$\mathcal L_V$ from 
the subclass ${\rm Sch}_{\mathbb R}$ is spanned by the vector fields
$D_{\mathbb R}(\tau)$, $G(\chi)$, $I$, $\sigma M$ and $Z(\eta^0)$,
where
\begin{gather*}
D_{\mathbb R}(\tau):=D(\tau)-\frac14\tau_tI=\tau\p_t+\frac12\tau_tx\p_x+
\frac18\tau_{tt}x^2M-\frac14\tau_tI,\quad
G(\chi)=\chi\p_x+\frac12\chi_txM,\\
M=i\psi\p_\psi-i\psi^*\p_{\psi^*},\quad 
I=\psi\p_\psi+\psi^*\p_{\psi^*},\quad
Z(\eta^0)=\eta^0\p_\psi+\eta^0{^*}\p_{\psi^*},
\end{gather*}
the parameters $\tau$, $\chi$ and $\sigma$ run 
through the set of real-valued 
smooth functions of $t$ satisfying the classifying condition~\eqref{classcondsub}, 
and $\eta^0$ runs through the solution set of the equation~$\mathcal L_V$.
\end{theorem}

It is obvious that properties of appropriate subalgebras for 
the subclass ${\rm Sch}_{\mathbb R}$ can be obtained 
by specifying the same properties of appropriate subalgebras for the whole class~\eqref{LinSchEqs}.
Thus, inequivalent cases of real-valued potentials admitting Lie symmetry 
extensions can be singled out from the classification list presented in Table~1. 
We note, however, that the group classification of real-valued potentials 
can be easily carried out from the outset.

\begin{theorem}\label{theoremresultsub}
A complete list of inequivalent Lie symmetry extensions in the subclass ${\rm Sch}_{\mathbb R}$ 
is given in Table~3.
\begin{table}[!th]
\renewcommand{\arraystretch}{1.92}
\begin{center}
Table 3. The classification list for real-valued potentials.
$$
\begin{array}{|c|c|c|c|l|}
\hline
\mbox{no.}&k_1&k_2& V &\hfil\mbox{Basis of }\mathfrak g^{\rm ess}_V\\
\hline
1 & 0 & 0 & V(t,x)& M,\;I \\
2 & 1 & 0 & V(x)& M,\;I,\;D(1)\\
3 & 3 & 0 & cx^{-2},\,c\in\mathbb R_*& M,\;I,\,D(1),\;D(t),\;D(t^2)-\frac12tI\\
4 & 3 & 2 & 0 & M,\;I,\;D(1),\;D(t),\;D(t^2)- \frac12 tI,\;G(1),\;G(t)\\
\hline
\end{array}
$$
\end{center}
\footnotesize
Lie symmetry extensions given in Table~3 are maximal if and only if
the potential~$V$ does not satisfy an equation of the form~\eqref{classcondsub} in~Case~1
and $V\ne b_2x^2+b_1x+b_0+c(x+a)^{-2}$ 
for any real constants~$a$, $b_0$, $b_1$, $b_2$ and $c$ in~Case~2. 
\end{table}
\end{theorem}

\begin{proof}
The proof follows the same pattern as Theorem~\ref{theoremresults}, and 
we sketch the proof by considering the invariants $k_1$ and $k_2$.
The case $k_2=0$ is split into the three subcases $k_1=0$, $k_1=1$ and $k_1\geqslant 2$. 
The proof for each subcase is the same as for Theorem~\ref{theoremresults} 
except that the parameter $\rho$ in each Lie symmetry vector field satisfies the equation $\rho_t=-\frac14\tau_{tt}$.
If $k_2=2$, then the algebra $\mathfrak g^{\rm ess}_V$ contains a vector field $Q^1=G(\chi^1)+\sigma^1M+\rho^1I$, 
where the parameters $\chi^1$ and $\sigma^1$ are real-valued smooth functions of $t$ with $\chi^1\ne 0$ 
and $\rho^1$ is a real constant.
Combining $Q^1$ with $I$ and using $G^\sim$-equivalence, we may assume that $Q^1=G(1)$.
The equation $\mathcal L_V$ is invariant with respect to $G(1)$ 
if and only if the potential $V$ does not depend on $x$. Then the equation $\mathcal L_V$
is equivalent to the free linear Schr\"odinger equation.
\end{proof}

\section{Conclusion}\label{Concl}

In this paper we have completely solved the group classification problem 
for (1+1)-dimensional linear Schr\"odinger equations with complex-valued potentials. 
The classification list is presented in Theorem~\ref{theoremresults} or, equivalently, 
in Corollary~\ref{mainresultsclassifiaction}.
This also gives the group classifications for the larger class of similar equations with variable mass 
and for the smaller class of such equations with real-valued potentials. 
We have introduced the notion of uniformly semi-normalized classes of differential equations 
and developed a special version of the algebraic method of group classification for such classes. 
This is, in fact, the main result of the paper. 
The class~\eqref{LinSchEqs} has the specific property 
of uniform semi-normalization with respect to linear superposition transformations, 
which is quite common for classes of homogeneous linear differential equations. 
Within the framework of the algebraic method, the group classification problem of the class~\eqref{LinSchEqs} 
reduces to the classification of appropriate low-dimensional subalgebras 
of the associated equivalence algebra~$\mathfrak g^\sim$. 

We show that the linear span~$\mathfrak g_\spanindex$ of the vector fields 
from the maximal Lie invariance algebras of equations from the class~\eqref{LinSchEqs}
is itself a Lie algebra. 
For each potential $V$, the maximal Lie invariance algebra~$\mathfrak g_V$ 
of the equation~$\mathcal L_V$ from the class~\eqref{LinSchEqs} is the semi-direct sum  
of a subalgebra $\mathfrak g^{\rm ess}_V$, of dimension not greater than seven,
and an infinite dimensional abelian ideal $\mathfrak g^{\rm lin}_V$, 
which is the trivial part of~$\mathfrak g_V$ and is associated with the linear superposition principle, 
$\mathfrak g_V=\mathfrak g^{\rm ess}_V\lsemioplus\mathfrak g^{\rm lin}_V$.
The above representation of~$\mathfrak g_V$'s yields 
a similar representation for~$\mathfrak g_\spanindex=\sum_V\mathfrak g_V$,  
$\mathfrak g_\spanindex=\mathfrak g^{\rm ess}_\spanindex\lsemioplus\mathfrak g^{\rm lin}_\spanindex$, 
where $\mathfrak g^{\rm ess}_\spanindex=\sum_V\mathfrak g^{\rm ess}_V$ 
is a (finite-dimensional) subalgebra of~$\mathfrak g_\spanindex$, 
and $\mathfrak g^{\rm lin}_\spanindex=\sum_V\mathfrak g^{\rm lin}_V$ is its abelian ideal.
The projection of the equivalence algebra~$\mathfrak g^\sim$ of the class~\eqref{LinSchEqs} on the space of variables 
coincides with~$\mathfrak g^{\rm ess}_\spanindex$.
Thus, two objects, $\mathfrak g^{\rm ess}_\spanindex$ and $\mathfrak g^\sim$, 
are directly related to the class~\eqref{LinSchEqs} and consistent with each other. 
This is why we classify appropriate subalgebras of~$\mathfrak g^{\rm ess}_\spanindex$ 
up to $G^\sim$-equivalence, each of which coincides with~$\mathfrak g^{\rm ess}_V$ for some~$V$.
\looseness=-1

The partition into classification cases is provided by two nonnegative integers $k_1$ and $k_2$, 
which are characteristic invariants of subalgebras of~$\mathfrak g^{\rm ess}_\spanindex$. 
This leads to two equivalent classification lists for the potential~$V$ depending on which of 
these invariants is assumed as the leading invariant. 
The list presented in Table~1 (resp.\ described in Corollary~\ref{mainresultsclassifiaction})
is constructed under the assumption that the invariant~$k_1$ (resp.\ $k_2$) is leading. 
Each of the lists consists of eight~$G^\sim$-inequivalent families of potentials. 
We have proved that for appropriate subalgebras the invariant~$k_2$ can take only two values: $0$ and $2$, 
and the invariant~$k_1$ is not greater than three. 
Further, the invariant $k_1$ cannot equal two for appropriate subalgebras due to the fact 
that the corresponding subalgebras cannot be maximal Lie symmetry algebras 
for equations from the class~\eqref{LinSchEqs}. 
At the same time, the proof of the condition $k_1\ne 2$ 
needs realizing the major part of the group classification of the class under study. 

The cases in the second list for which $k_2=0$ coincide with those from the first list.
For $k_2=2$, the group classification of the class~\eqref{LinSchEqs} reduces to 
the group classification of its subclass~\eqref{igammax subclass}.
This subclass is uniformly semi-normalized with respect to a larger family of point symmetry groups 
than the corresponding groups of linear superposition transformations, 
which makes the subclass~\eqref{igammax subclass} a useful example for group analysis of differential equations. 
For each equation~$\mathcal L_V$ from the subclass,
the essential part $\mathfrak g^{\rm ess}_V$ of its maximal Lie invariance algebra~$\mathfrak g_V$ 
can be written as 
$\mathfrak g^{\rm ess}_V=\mathfrak g^{\rm ext}_V\lsemioplus(\mathfrak g^{\rm unf}_V\cap\mathfrak g^{\rm ess}_\spanindex)$, 
where $\mathfrak g^{\rm unf}_V$ is an ideal of $\mathfrak g_V$ and has a similar structure for all equations from the subclass,
and $\mathfrak g^{\rm ext}_V$ is a subalgebra of $\mathfrak g^{\rm ess}_V$.
The vector fields from all $\mathfrak g^{\rm ext}_V$'s of equations from the subclass~\eqref{igammax subclass} 
constitute the~algebra $\mathfrak g^{\rm ext}_\spanindex$, 
which is contained in the projection of the equivalence algebra 
of the subclass~\eqref{igammax subclass} and is isomorphic to the algebra ${\rm sl}(2,\mathbb R)$. 
Therefore, the classification of subalgebras of ${\rm sl}(2,\mathbb R)$ (which is well known) yields
the solution of the group classification problem of the subclass~\eqref{igammax subclass}, 
whose result is presented in Table~2.
  
Since the subclass ${\rm Sch}_{\mathbb R}$ of (1+1)-dimensional linear Schr\"odinger equations with real-valued potentials  
is important for applications, we have given its group classification separately 
by singling out related results from the group classification of the class~\eqref{LinSchEqs}.
Since the subclass ${\rm Sch}_{\mathbb R}$ is also uniformly semi-normalized with respect to linear superposition of solutions, 
this procedure can be realized within the framework of the algebraic approach 
by specifying the properties of appropriate subalgebras for the case of real-valued potentials.

Furthermore, the semi-normalization of the above classes of linear Schr\"odinger equations guarantees that 
there are no additional point equivalence transformations between classification cases listed for each of these classes. 

The new version of the algebraic method
that is given in Section~\ref{SectionOnUniformlySemi-normalizedClasses}  
and then applied to the symmetry analysis of the class~\eqref{LinSchEqs}
can be regarded as a model 
for optimizing the group classification of other classes of differential equations (including higher-dimensional cases).
We intend to extend our approach 
to multidimensional linear Schr\"odinger equations with complex-valued potentials. 
In this context, it seems that the technique used in the proof of Theorem~\ref{theoremresults}
is more useful for generalizing to the multidimensional case than the alternative proof presented in Section~\ref{LinScEqsalternative}.

\subsection*{Acknowledgements}

The research of C.K. was supported by International Science Programme (ISP) in collaboration 
with East African Universities Mathematics Programme (EAUMP).
The research of R.O.P. was supported by the Austrian Science Fund (FWF), project P25064.
The authors are pleased to thank Anatoly Nikitin, Olena Vaneeva and Vyacheslav Boyko for stimulating discussions.


\begin{thebibliography}{99}\itemsep=0ex\footnotesize\frenchspacing

\bibitem{Anderson&Kumei&Wulfman1972a}
Anderson R.L., Kumei S. and Wulfman C.E., 
Invariants of the equations of wave mechanics. I, 
{\it Rev. Mexicana F\'\i s.} {\bf 21} (1972), 1--33.  

\bibitem{Anderson&Kumei&Wulfman1972b}
Anderson R.L., Kumei S. and Wulfman C.E., 
Invariants of the equations of wave mechanics. II. One-particle Schroedinger equations, 
{\it Rev. Mexicana F\'\i s.} {\bf 21} (1972), 35--57. 

\bibitem{Basarab&Lahno&Gungor}
Basarab-Horwath P., G\"ung\"or\ F. and Lahno V., 
Symmetry classification of third-order nonlinear evolution equations. Part I: Semi-simple algebras, 
{\it Acta Appl. Math.} {\bf 124} (2013), 123--170.

\bibitem{Basarab&Lahno&Zhadonov}
Basarab-Horwath P., Lahno V. and Zhdanov R.,
The structure of Lie algebras and the classification problem for partial differential equations,
{\it Acta Appl. Math.} {\bf 69} (2001), 43--94.

\bibitem{Bihlo&Cardoso-Bihlo&Popovych2011}
Bihlo A., Dos Santos Cardoso-Bihlo E. and Popovych R.O., 
Enhanced preliminary group classification of a class of generalized diffusion equations, 
{\it Commun. Nonlinear Sci. Numer. Simul.} {\bf 16} (2011), 3622--3638, arXiv:1012.0297. 

\bibitem{Bihlo&Cardoso-Bihlo&Popovych2012}
Bihlo A., Dos Santos Cardoso-Bihlo E. and Popovych R.O., 
Complete group classification of a class of nonlinear wave equations, 
{\it J.~Math. Phys.} {\bf 53} (2012), 123515, 32 pp., arXiv:1106.4801.

\bibitem{Bihlo&Popovych2017}
Bihlo A. and Popovych R.O., 
Group classification of linear evolution equations, 
{\it J.~Math. Anal. Appl.} {\bf 448} (2017), 982--1005, arXiv:1605.09251.

\bibitem{Bluman&Kumei}
Bluman G.W. and Kumei S.,
{\it Symmetries and differential equations}, 
Springer, New York, 1989. 

\bibitem{Bender2004} 
Bender C.M., Complex Extension of Quantum Mechanics, 
in {\it Proceedings of Fifth International Conference ``Symmetry in Nonlinear Mathematical Physics'' (23--29 June, 2003, Kyiv)},
{\it Proceedings of Institute of Mathematics of NAS of Ukraine} {\bf 50} (2004), Institute of Mathematics of NAS of Ukraine, Kyiv, pp. 617--628.

\bibitem{Bender2007} 
Bender C. M., 
Making sense of non-Hermitian Hamiltonians, 
{\it Rep. Progr. Phys.} {\bf 70} (2007), 947--1018. 

\bibitem{Boyer1974}
Boyer C.P.,
The maximal `kinematical' invariance group for an arbitrary potential,
{\it Helv. Phys. Acta} {\bf 47} (1974), 589--605.

\bibitem{Doebner&Goldin1994}
Doebner H.-D. and Goldin G.A.,
Properties of nonlinear Schr\"odinger equations associated with diffeomorphism group representations, 
{\it J.~Phys. A} {\bf 27} (1994), 1771--1780.

\bibitem{Cardoso-Bihlo&Bihlo&Popovych2011}
Dos Santos Cardoso-Bihlo E., Bihlo A. and Popovych R.O.,  
Enhanced preliminary group classification of a class of generalized diffusion equations, 
{\it Commun. Nonlinear Sci. Numer. Simul.} {\bf 16} (2011), 3622--3638, arXiv:1012.0297. 

\bibitem{Fushchich&Moskaliuk1981}
Fushchich W.I. and Moskaliuk S.S.,
On some exact solutions of the nonlinear Schr\"odinger equation in three spatial dimensions, 
{\it Lett. Nuovo Cimento (2)} {\bf 31} (1981), 571--576.

\bibitem{Gagnon88a}
Gagnon L. and Winternitz~P.,
Lie symmetries of a generalised non-linear Schr\"odinger equation. I. The symmetry group and its subgroups, 
{\it J.~Phys.~A} {\bf 21} (1988), 1493--1511.

\bibitem{Gagnon89a}
Gagnon L. and Winternitz~P.,
Lie symmetries of a generalised non-linear Schr\"odinger equation. II. Exact solutions, 
{\it J.~Phys.~A} {\bf 22} (1989), 469--497.

\bibitem{Gagnon89b}
Gagnon L., Grammaticos B., Ramani A. and Winternitz P.,
Lie symmetries of a generalised non-linear Schr\"odinger equation. III. Reductions to third-order ordinary differential equations, 
{\it J.~Phys.~A} {\bf 22} (1989), 499--509.

\bibitem{Gagnon89c}
Gagnon L. and Winternitz P.,
Exact solutions of the cubic and quintic nonlinear Schr\"odinger equation for a cylindrical geometry, 
{\it Phys. Rev.~A~(3)} {\bf 39} (1989), 296--306.

\bibitem{Gagnon93}
Gagnon L. and Wintenitz P., 
Symmetry classes of variable coefficient nonlinear Schr\"odinger equations, 
{\it J.~Phys.~A} {\bf 26} (1993), 7061--7076.

\bibitem{Gazeau92}
Gazeau J.P. and Winternitz P., 
Symmetries of variable coefficient Korteweg--de Vries equations, 
{\it J.~Math. Phys.} {\bf 33} (1992), 4087--4102.

\noprint{
\bibitem{Ivanova2002}
Ivanova N., 
Symmetry of nonlinear Schr\"odinger equations with harmonic oscillator type potential, 
{\it Symmetry in nonlinear mathematical physics, Part 1, 2 (Kyiv, 2001)}, 
{\it Pr. Inst. Mat. Nats. Akad. Nauk Ukr. Mat. Zastos.} {\bf 43}, 
Part 1, 2, Nats\=\i onal. Akad. Nauk Ukra\"\i ni, \=Inst. Mat., Kyiv, pp. 149--150. 
}

\bibitem{Lie1881withTrans}
Lie S., 
\"Uber die Integration durch bestimmte Integrale von einer Klasse linear partieller Differentialgleichung, 
{\it Arch. for Math.} {\bf 6} (1881), 328--368.
(Translation by N.H. Ibragimov:
Lie S., 
On integration of a class of linear partial differential equations by means of definite integrals, 
{\it CRC Handbook of Lie Group Analysis of Differential Equations}, vol.~2, CRC Press, Boca Raton, 1994, pp. 473--508.)

\bibitem{Lisle1992}
Lisle~I.G.,
{\it Equivalence transformations for classes of differential equations}, 
PhD. thesis, University of British Columbia, 1992.

\bibitem{Magadeev1993}
Magadeev B.A., 
Group classification of nonlinear evolution equations,
{\it Algebra i Analiz} {\bf 5} (1993), 141--156 (in Russian); English
  translation in {\it St.~Petersburg Math.~J.} {\bf 5} (1994), 345--359.

\bibitem{Miller1977}
Miller W., 
{\it Symmetry and separation of variables}, 
Addison-Wesley Publishing Co., Reading, MA, 1977. 

\bibitem{Mostafazadeh2014}
Mostafazadeh A., 
A dynamical formulation of one-dimensional scattering theory and its applications in optics, 
{\it Ann. Physics} {\bf 341} (2014), 77--85.

\bibitem{Nattermann&Doebner1996}
Nattermann P. and Doebner H.-D.,
Gauge classification, Lie symmetries and integrability of a family of nonlinear Schr\"odinger equations, 
{\it J. Nonlinear Math. Phys.} {\bf 3} (1996), 302--310. 

\bibitem{Niederer1972}
Niederer U., 
The maximal kinematical invariance group of the free Schr\"odinger equation,
{\it Helv. Phys. Acta.} {\bf 45} (1972), 802--810.

\bibitem{Niederer1973a}
Niederer U., 
The maximal kinematical invariance group of the harmonic oscillator,
{\it Helv. Phys. Acta.} {\bf 46} (1973), 191--200.

\bibitem{Niederer1973b}
Niederer U., 
The group theoretical equivalence of the free particle, the harmonic oscillator and the free fall, 
in {\it Proceedings of the 2nd International Colloquium on Group Theoretical Methods in Physics}, University of Nijmegen, The Netherlands, 1973.

\bibitem{Niederer1974}
Niederer U., 
The maximal kinematical invariance groups of Schr\"odinger equations with arbitrary potentials,
{\it Helv. Phys. Acta.} {\bf 47} (1974), 167--172.

\bibitem{Olver1993}
Olver P.J., 
{\it Applications of Lie groups to differential equations}, 
 Springer-Verlag, New York, 1993. 
 
\bibitem{Opanasenko&Bihlo&Popovych2017}
Opanasenko S., Bihlo A. and Popovych R.O., 
Group analysis of general Burgers--Korteweg--de Vries equations, 
{\it J.~Math. Phys.} {\bf 58} (2017), 081511, 37 pp., arXiv:1703.06932.

\bibitem{Ovsiannikov1959}
Ovsiannikov L.V., Group properties of nonlinear heat equation
{\it Dokl. AN SSSR} {\bf 125} (1959), 492--495 (in Russian).

\bibitem{Ovsiannikov1982}
Ovsiannikov L.V., 
{\it Group analysis of differential equations}, 
Academic Press, New York, 1982.

\bibitem{Ovsiannikov&Ibragimov1975}
Ovsjannikov L.V. and  Ibragimov N. H., 
Group analysis of the differential equations of mechanics,
in {\it General mechanics}, vol. 2, Akad. Nauk SSSR Vsesojuz. Inst. Nauchn. i Tehn. Informacii, Moscow, 1975, pp. 5--52 (in Russian).

\bibitem{Popovych&Ivanova2003a}
Popovych R.O. and Ivanova N.M.,
New results on group classification of nonlinear diffusion-convection equations, 
{\it J.~Phys.~A} {\bf 37} (2004), 7547--7565, arXiv:math-ph/0306035.

\bibitem{Popovych&Ivanova&Eshraghi2003CubicLanl}
Popovych R.O. and Ivanova N.M. and Eshraghi H.,
Lie symmetries of (1+1)-dimensional cubic Schr\"odinger equation with potential, 
{\it Proceedings of Institute of Mathematics of NAS of Ukraine} {\bf 50} (2004), 219--224, arXiv:math-ph/0310039.
 
\bibitem{Popovych&Ivanova&Eshragi2004}
Popovych R.O., Ivanova N.M and Eshragi H., 
Group classfication of (1+1)-dimensional Schr\"odinger equations with potentials and power nonlinearities, 
{\it J.~Math. Phys} {\bf 45} (2004), 3049--3057, arXiv:math-ph/0311039.

\bibitem{Popovych&Kunzinger&Eshragi2010}
Popovych R.O., Kunzinger M. and Eshragi H., 
Admissible transformations and normalized classes of non-linear Schr\"odinger equations, 
{\it Acta Appl. Math.} {\bf 109} (2010), 315--359, arXiv:math-ph/0611061. 

\bibitem{Popovych&Kunzinger&Ivanova2008}
Popovych R.O., Kunzinger M. and Ivanova N.M., 
Conservation laws and potential symmetries of linear parabolic equations, 
{\it Acta Appl. Math.} {\bf 100} (2008), 113--185, arXiv:0706.0443. 

\noprint{
\bibitem{Srivastava&Bose2009}
Srivastava V.K. and Bose S.K., 
Exact solution of relativistic Schr\"odinger equation for the central complex potential $V(r) = iar+\frac br$, 
{\it Indian Journal of Pure and Applied Physics} {\bf 47} (2009), 547--550.
}

\bibitem{Vaneeva&Popovych&Sophocleous2009}
Vaneeva O.O., Popovych R.O. and Sophocleous C.,
Enhanced group analysis and exact solutions of variable coefficient semilinear diffusion equations with a power source,
{\it Acta Appl. Math.} {\bf 106} (2009), 1--46, arXiv:0708.3457.

\bibitem{Vaneeva&Popovych&Sophocleous2012}
Vaneeva O.O., Popovych R.O. and Sophocleous C.,  
Extended group analysis of variable coefficient reaction-diffusion equations with exponential nonlinearities, 
{\it J. Math. Anal. Appl.} {\bf 396} (2012), 225--242, arXiv:1111.5198. 

\bibitem{Zhdanov&Roman2000}
Zhdanov R. and Roman O., 
On preliminary symmetry classification of nonlinear Schr\"odinger equations with some applications to Doebner--Goldin models, 
{\it Rep. Math. Phys.} {\bf 45} (2000), 273--291. 

\bibitem{Zhdanov&Lahno1999}
Zhdanov R.Z. and Lahno V.I.,
Group classification of heat conductivity equations with a nonlinear source, 
{\it J.~Phys.~A} {\bf 32} (1999), 7405--7418.

\end{thebibliography}
\end{document}